\documentclass[10pt, conference, letterpaper]{IEEEtran}
\ifCLASSINFOpdf
\else
\fi
\usepackage{amssymb,amsmath, amsthm, latexsym, epstopdf,graphics, color}
\usepackage{todonotes}
\usepackage{tikz}
\usetikzlibrary{arrows}
\usepackage{fancyhdr}
\usepackage{subfig}
\usepackage{amsmath}
\usepackage{graphicx}
\usepackage{epstopdf}

\usepackage{makecell}

\theoremstyle{definition}
\newtheorem{definition}{Definition}

\newtheorem{example}{Example}

\newtheorem{proposition}{Proposition}

\begin{document}
%
\title{A Mood Value for Fair Resource Allocations}



%
\author{\IEEEauthorblockN{Francesca Fossati\IEEEauthorrefmark{1},
Stefano Moretti\IEEEauthorrefmark{2},
Stefano Secci\IEEEauthorrefmark{1}
}
\IEEEauthorblockA{\IEEEauthorrefmark{1}Sorbonne Universit\'{e}s, UPMC Univ Paris 06, UMR 7606, LIP6, 75005 Paris, France. Email: \{firstname.lastname\}@upmc.fr}
\IEEEauthorblockA{\IEEEauthorrefmark{2} CNRS UMR7243, PSL, Universit\'{e} Paris-Dauphine, Paris, France. Email: stefano.moretti@lamsade.dauphine.fr.}}


\maketitle

\begin{abstract}
	In networking and computing,
	resource allocation is typically addressed using classical sharing protocols as, for instance, the proportional division rule, the max-min fair allocation, or other solutions inspired by cooperative game theory.
	In this paper, we argue that, under awareness about the available resource and other users' demands, in a cooperative setting such classical resource allocation approaches, as well as associated notions of fairness, show important limitations. 
	We identify in the individual satisfaction rate the key aspect of the challenge of defining a new notion of fairness and, consequently, a resource allocation algorithm more appropriate for the cooperative context.  We generalize the concept of user satisfaction 
	considering the set of admissible solutions for bankruptcy games. We adapt the Jain's fairness index to include the new user satisfaction rate. Accordingly, we propose  a new allocation rule we call  \lq Mood Value\rq. For each user it equalizes our novel game-theoretic definition of user satisfaction with respect to a distribution of the resource. We test the mood value and the new fairness index through extensive simulations
	showing how they better support the fairness analysis.  
\end{abstract}
\section{Introduction}

In communication networks and computing systems, resource allocation (in some contexts also referred to as resource scheduling, pooling, or sharing) is a phase, in a network protocol or system management stack, when a group of individual users or clients have to receive a portion of the resource in order to operate a service. Resource allocation becomes a challenging problem when the available resource is limited and not enough to fully satisfy users' demand. In such situations, resource allocation algorithms need to ensure a form of fairness. Such situations emerge in a variety of contexts, such as wireless  access~\cite{saad-debbah,cath}, competitive routing~\cite{orda}, transport control~\cite{proutiere-server}. 

The common methodology adopted in the literature is to, on the one hand, determine allocation rules such that they satisfy desirable properties~\cite{thomson}, and, on the other hand, analyse the fairness of a given allocation through indices, the most commonly used being the Jain's index~ \cite{jain}. 
Allocation rules and indices of fairness are commonly justified by some fairness criteria. For instance, among two equivalent users demanding the same amount of resource, it makes sense not to discriminate and to give to each of them the same portion of the resource. In some cases, it can be desirable to guarantee at least a minimum amount of the resource so that the maximum number of users can be served. 

In the networking literature, the resource allocation problem is historically solved as a single-decision maker problem in which users are possibly not aware of the other users' demands and of the total amount of available resource. It follows that the most natural and intuitive way to quantify the user satisfaction is through the proportion of the demand that is satisfied by an allocation. Large literature exists indeed in the networking area on proportional resource allocations for many practical situations, from wireless networks to transport connection management~\cite{ cath,orda,proutiere-server}. 

In this paper, we are particularly interested instead in cooperative networking contexts such that users can be aware of other users' demands and the available amount. As such, rational users shall compute their satisfaction also based on the presence of other users. In fact, such networking contexts with demand and resource availability awareness are making surface in wired and wireless network environments with an increasing level of programmability, i.e., using software-defined radio and network platforms that expose novel (northbound) interfaces to users to disseminate information and pilot network resource allocations. Our main idea is defining a new notion of user satisfaction for such interactive resource allocation situations with demand and resource awareness. 

	\begin{figure}[t] 
		\centering
		\includegraphics[scale=0.3]{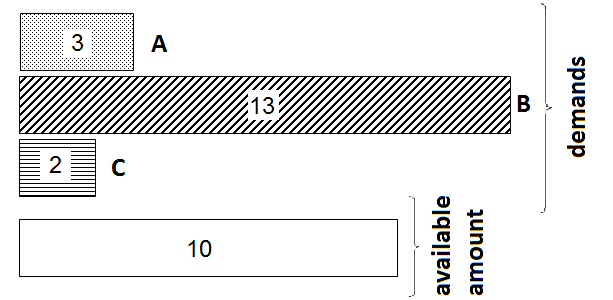}
		\caption{A critical resource allocation situation example}
        \vspace{-0.5cm}
		\label{es}
	\end{figure}
	
Let us briefly clarify our motivation with the following allocation example.  
A user $i$ asks a quantity of resource that is bigger than the resource itself (as B in Fig. \ref{es}). Classical fairness indices~\cite{jain},~\cite{atk},~\cite{alpha} tend to qualify the user satisfaction as maximum when $i$ obtains exactly what it asks. In the case where $i$ asks more than the available amount, it cannot reach the maximum satisfaction due to the fact that its demand exceeds the available resource. 
Instead, in demand and resource awareness conditions, it would be more reasonable that its satisfaction is maximum when it obtains all the available resource. Furthermore, if all the other users together ask a quantity of good inferior to the resource, a minimum portion of it, equal to the difference between the resource and the sum of the demands of all the others, is guaranteed to $i$. Under a dual reasoning, it also appears more acceptable that the minimum satisfaction of a user is reached when it receives the minimum portion of the available resource, instead of when it receives zero.
If users are in complete information context the classical approach can lead to not reasonable outcomes.

In this perspective, in order to better describe the user satisfaction as a function of the available resource, and to capture the interactions due to the networking context (e.g., networked users may be aware of respective demands, may ally in the formulation of their demands, etc), we propose to model the resource allocation problem as a coalitional game. Accordingly, we define a new satisfaction rate for users, able to adapt to various configurations of the demands. Furthermore, we define a new resource allocation rule, called the \lq Mood Value\rq, based on the idea that the most fair allocation is the one that equalizes the satisfaction of each player.  Indeed, regardless of the level of satisfaction, each player is not discriminated if its satisfaction is the same than the one of all the others. Choosing this allocation, users, who have the chance to recover informations about the other users and the available resource, have the feeling to receive a fair portion of the resource. We also provide an interpretation of this approach positioning it with respect to classical traffic theory~\cite{prop}.

The paper is organized as follows. Section~\ref{background} presents the state of the art on the topic. In Section~\ref{satisf} a new satisfaction rate is proposed. In Section~\ref{moodvalue} the mood value and a new fairness index are described. In Section~\ref{traffictheory} we provide an interpretation of the mood value with a traffic theory methodology, i.e. as result of the maximization of an appropriate utility function. Section~\ref{results} presents some numerical examples. Finally, Section~\ref{concl} concludes the paper.

\section{Background}
\label{background}

A resource allocation problem can be characterized by a pair $(c, E)$, in which $c$ is the vector of demands (claims) from $n$ users (claimants) and $E$ is the resource (estate) that should be shared between them. The set of users is $N=\{1,...,n\}$. The resource allocation is a challenging problem when $E$ is not enough to satisfy all the demands ($\sum\limits_{i=1}^n c_i \geq E$). 
An allocation $x \in \mathbb R^n$ is a solution vector that satisfies three basic properties:
\begin{itemize}
	\item \textit{Non-negativity}: each user should receive at least zero.
	\item  \textit{Demands boundedness}: each user cannot receive more than its demand.
	\item \textit{Efficiency}: the sum of all allocations should be $E$.
\end{itemize} An allocation rule is a function that associates a unique allocation vector $x$ to each $(c ,E)$.

\subsection{Classical resource allocation rules}
Many resource allocation rules are proposed in the literature and each of them is characterized by a set of properties that justify the use of the given rule in order to find a solution of the allocation problem~\cite{thomson}. In computer networks, the most well-known rules are: 
the proportional rule and the weighted proportional rule~\cite{prop}, the max-min fair allocation (MMF)~\cite{mmf},~\cite{mmf1}
,  and the $\alpha$-fair allocation~\cite{alpha}. Each of these allocation rules, result of an optimization problem and/or an iterative algorithm, follows a fairness criterion.  

The  \textit{weighted proportional allocation rule} is based on the idea that a logarithmic utility function captures well the individual evaluation of the worth of the resource~\cite{prop}. One way to compute it is via the maximization of $\sum\limits_{i=1}^{n} w_i\log x_i$ subject to demand boundness and efficiency constraints. When  
$w_i$ is equal to $1$ the resulting allocation is called simply proportional and when $w_i$ is equal to $c_i$ we obtain the allocation that actually produces allocations proportional to the demands; hence in the following, we refer to the latter rule as \lq proportional\rq \ instead of the previous (not weighted) one.  

The idea behind the \textit{max-min fairness (MMF) allocation} is to maximize firstly the minimum allocation; secondly, the second lowest allocation, and so on~\cite{mmf,mmf1}. This solution coincides with the only feasible allocation such that, if the allocation of some users is increased, the allocation of some other users with smaller or equal amount is decreased. 

More generally, it is possible to obtain a family of allocation rules maximizing a parametric utility function. The \textit{$\alpha$-fair utility} function is defined as $\sum\limits_{i=1}^n \frac{x_i^{(1-\alpha)}}{1-\alpha}$ \cite{alpha}. If $\alpha\to 1$ the solution of the optimization problem coincides with the weighted proportional allocation with $w_i$ equal to $1$, if $\alpha=2$ with the minimum delay potential allocation, that is the allocation obtained minimizing the total potential delay $\sum\limits_{i=1}^{n} ( \frac{1}{x_i}) $ \cite{delay}, and if $\alpha\to \infty$ with the max-min fair allocation. 

\subsection{Game theoretical allocation rules}

Recently game theory has been applied to communication systems in order to model network interactions. For example, in \cite{cop} a cooperative game model is proposed to select a fair allocation of the transmission rate in multiple access channels and in \cite{cop2} the authors studied, using coalitional game theory, the cooperation between rational users in wireless networks.
 
Moreover, it is possible to analyze the allocation problem as a Transferable Utility (TU) game \cite{net,nucleolus,bank}, which is defined as a pair $(N,v)$, where $N=\{1,\dots,n\}$ denotes the set of \textit{players} and $v:2^{N}\to \mathbb{R}$
is the {\it characteristic function},
(by convention, $v(\emptyset)=0$). Bankruptcy games~\cite{thomson}, in particular, deal with situations where the number of claimed resource exceeds that available. A \textit{Bankruptcy game} is a TU-game $(N,v)$ in which the value of the coalition is given by 
\begin{equation}
{v(S)= \max\{E-\sum_{i \in N \setminus S} c_i, 0\}}
\end{equation}
where $E \geq 0$ represents the estate to be divided and $c \in \mathbb{R}_+^N$ is a vector of claims satisfying the condition $\sum_{i \in N} c_i > E$ \cite{game,aumann}.
The bankruptcy game is \textit{superadditive}, that is: 
\begin{equation} 
{v(S \cup T) \geq v(S) + v(T)}, \ \ \ \forall S, T \subseteq N \vert S \cap T = \emptyset 
\end{equation} 
 it is also \textit{supermodular} (or, equivalently, \textit{convex}), that is:
 \begin{equation}
 { v(S \cup T) + v(S \cap T)  \geq v(S) + v(T) } \ \forall S, T \subseteq N
 \end{equation} 

A classical set-value solution for a TU-game is the \textit{core} $C(v)$, which is is defined as the
set of \textit{allocation vectors} $x \in \mathbb{R}^N$ for which no coalition has an
incentive to leave the grand coalition $N$,~i.e.:
\begin{equation}{
	C(v)=\{ x \in \mathbb{R}^N: \sum_{i \in N} x_i = v(N),  \sum_{i \in S} x_i \geq v(S) \quad \forall S \subset N\}.} \end{equation}


%
A \emph{one-point solution} (or simply a \emph{solution}) for a
class $\mathcal{C}^N$ of coalitional games is a function $\psi:
\mathcal{C}^N \rightarrow \mathbb{R}^{N}$ that assigns a payoff
vector $\psi(v) \in \mathbb{R}^{N}$ to every coalitional game in the
class. A well-known solution for TU-games is the \textit{Shapley value} \cite{shapley} $\phi(v)$ of a game $(N,v)$, defined as the weighted mean of the players' marginal contributions over all possible coalitions and computed as follows:
\begin{equation}
\phi_i(v)= \sum_{S \subseteq N: i \in S} w_i(S)
(v(S)-v(S \setminus \{i\})),
\end{equation}
 with $ w_i(S)=\frac{(s-1)!(n-s)!}{n!} $ where $s$ denotes the cardinality of $S \subseteq N$.
%
Another well studied solution for TU-games is the nucleolus, based on the idea of minimizing the maximum discontent \cite{nucl}.
Given a TU-game $(N,v)$ and an allocation $x \in \mathbb{R}^N$, let $e(S,x)= v(S) - \sum_{i \in S} x_i$ be the \textit{excess} of coalition $S$ over the allocation $x$, and let $\leq_L$ be the {\it lexicographic} order on $\mathbb{R}$. Given an imputation $x$, $\theta(x)$ is the vector that  arranges  in  decreasing  order  the  excess of  the $2^n-1$ non-empty coalitions  over  the  imputation
$x$.
The \textit{nucleolus} $\nu(v)$ is defined as the  \textit{imputation}  $x$ (i.e., $\sum_{i \in N} x_i = v(N)$ and $x_i \geq v(\{i\})$ for each $i \in N$) such that $\theta(x) \leq_L \theta(y)$ for all $y$ imputations of the game $v$.

 Given a bankruptcy game, many other solutions can be proposed \cite{thomson}. As already introduced in the previous section, the \textit{proportional allocation} assigns to player  $i$ an allocation equal to $E \cdot c_i /  \sum\limits_{i=1}^n c_i$.
   For example, it is worth mentioning the  \textit{Constrained Equal Loss} (CEL) allocation that divides equally the difference between the sum of the demands and $E$, under the constraint that no player receives a negative
amount.
   
%
%

\subsection{Fairness indices}

The evaluation of the fairness of the allocations, used as an important system performance metric especially in networking, can be useful to discriminate among allocation rules and to evaluate the level of \lq justice\rq\ in the  repartition of the resources. 
Jain~\cite{jain} introduces a formula aimed at providing a quantitative measure of  the  fairness  of  a resource  sharing  allocation. 
\begin{definition}[Jain's index]
	Given an allocation problem $ (c,E) $ and an allocation $  x $, the \textit{Jain's fairness index} is:\begin{equation}
	\label{jain}
	{J=
	\biggl[\sum\limits_{i=1}^n \bigl(\frac{x_i}{c_i}\bigr) \biggr]^2 
	 \bigg/\biggl[ n\sum\limits_{i=1}^n \bigl(\frac{x_i}{c_i}\bigr)^2\biggr] } 
	\end{equation} 
\end{definition}

The Jain's index is bounded between $ \frac{1}{n} $ and $1$ \cite{jain}. The maximum fairness is measured when all the users obtain the same fraction of demand and the minimum fairness is measured when it exists only one user that receives all the resource. 
The Jain's index has the following good properties: 
\begin{itemize}
	\item \textit{Population size independence}: applicable to any user set, finite or infinite.
	\item \textit{Scale and metric independence}: not affected by the scale.
	\item \textit{Boundedness}: can be expressed as a percentage.
	\item\textit{Continuity}: able to capture any change in the allocation.
\end{itemize}

The index considers the proportion of demand and it gives the maximum fairness to the allocation for which all the users receive the same proportion of the demand, regardless of the type of allocation problem, it suggests to allocate the resources in a proportional way even when this allocation is not the most suitable to solve the problem. 
Another well-know index of fairness is the Atkinson's index \cite{atk}; contrary to the Jain's index, it measures the degree of inequality of a given allocation, taking value equal to 0 when the system is $100\%$ fair in the MMF sense, and 1 when it is totally unfair.



\begin{example}

	Let $(c,E)$ be the situation of Fig.~\ref{es} with $ c=(3,13,2)$ and $ E=10 $. The discussed allocation rules provide values in Table~\ref{tab1} along with the Jain's index and 1-Atkinson's index in order to have a measure of fairness. 

\begin{table}[h!]
   \vspace{-0.25cm}
	\footnotesize
	\begin{center}
		\begin{tabular}{c|c|c|c|c|c}
			\rule[-6pt]{0mm}{18pt}
			User demands &Prop. &MMF& Shapley &Nucleolus & CEL\\
			\hline
			\rule[-6pt]{0mm}{18pt}
			A: 3 &1.67&3&1.5&1&0\\
			\rule[-6pt]{0mm}{10pt}
			B: 13&7.22&5&7.5&8&10\\
			\rule[-6pt]{0mm}{10pt}
			C: 2&1.11&2&1&1&0\\
			\hline			
			Jain's index& 1& 0.882& 0.995&0.946&0.333\\
			Atkinson's index & 0.844 & 0.965 & 0.821 & 0.777 & 0.333\\
		\end{tabular}		
		\caption{Allocation rules: comparison ($E=10$, cf. Fig. \ref{es}).}
        \vspace{-0.5cm}
		\label{tab1}
	\end{center}	
\end{table}
\end{example}
%
%

The axiomatic theory of fairness proposed in~\cite{axiom} shows that it exists an unique family of fairness measures, which includes the Jain's and the Atkinson's indices, satisfying a set of reasonable axioms.
In the rest of the paper, we  consider only the Jain's index because it is the one classically used in networking applications. MMF-driven inequality indices find their most appropriate use in socio-economical contexts, because they are linked to the concept of welfare of an income distribution. Furthermore the Jain's index is based on the idea of summarizing the information about the users' satisfaction, which is close to our methodology of redefining users' satisfaction under demand and resource awareness, as discussed in the following section. 


\section{From demand fraction satisfaction to game theoretical satisfaction}
\label{satisf}


In this section, we propose a game-theoretic approach to evaluate the satisfaction of a user for an allocation. 

\subsection{User satisfaction rate}

A crucial issue in resource allocation is to jointly:

\begin{itemize}
	\item find the best solution in terms of a certain goal;
	\item evaluate its fairness by referring to a fairness index.
\end{itemize}
With this purpose, it is important to evaluate the individual satisfaction rates and to summarize the information given by each of them with a global fairness index.

A natural way to quantify the satisfaction of a user, as proposed by Jain, is through the proportion of the demand that is satisfied by an allocation \cite{jain}.  

\begin{definition}[Demand Fraction Satisfaction rate]
	Given the user $i$ with demand $c_i$ and an allocation $ x_i $, the \textit{Demand Fraction Satisfaction (DFS) rate} of $i$ is:
	\begin{equation}\label{fd}	
	{
		DFS_i =\frac{x_i}{c_i}.}
	\end{equation}\end{definition}
This rate takes value between 0 and 1 since it represents the percentage of the demand that is satisfied.

Unavoidably, this way to quantify the user satisfaction makes the weighted proportional allocation the fairest one since it allocates proportionally to the demand. There are, however, situations in which the common sense does not suggest to allocate in a proportional way; e.g., if there is a big gap between the demands, in order to protect the 
\lq weaker\rq \ users and guarantee them a minimum portion of the estate, the MMF allocation can be preferable. Furthermore, as mentioned in the introduction, the presence of other users should rationally be considered not to distort the satisfaction of each user, in case of awareness about other users' demand and the available demands. 


For these reasons, we aim at defining an alternative satisfaction rate such that it satisfies the following two properties we name demand relativeness and relative null satisfaction:
\begin{itemize}
	\item \textit{Demand relativeness}: a user is fully satisfied when it receives its maximal right, based on the available resource;
	\item \textit{Relative null satisfaction}: a user has null satisfaction when it receives exactly its minimal right, based on other users' demands and the available resource.
\end{itemize}
The minimal right for a player is the difference between the available amount and the sum of the demands of the other users (i.e., taking a worst case assumption that the others get the totality of their demand), and the maximal right is equal to the maximum available resource, i.e., $c_i$ if $c_i<E$, or it is equal to E otherwise. Remembering the definition of the characteristic function of a bankruptcy game we have that:
\begin{itemize}
	\item the \textit{minimal right} for player $i$ is  $v(i)$
	\item the \textit{maximal right} for player $i$ is  $v(N)-v(N\setminus i)$
\end{itemize}

Thus we introduce the \lq player satisfaction (PS) rate\rq, which  satisfies the above two properties by considering the value of the bankruptcy game associated to the allocation problem.

\begin{definition}[Player Satisfaction Rate]	
	Given a bankruptcy game such that $ \sum\limits_{i=1}^n c_i>E $ and an allocation $ x_i $, the {\it{Player Satisfaction (PS) rate}} for $i$ is:
	\begin{equation}\label{player}{
		PS_i=\dfrac{x_i-min_i}{max_i-min_i},}
		\end{equation}
	where: $min_i=v(i)$, $max_i=v(N)-v(N\setminus i)$.
	If $ \sum_{i=1}^n c_i=E $ the player satisfaction rate is $ PS_i=1$, $ \forall i \in N$. 
\end{definition}

$PS_i \in [0,1]$ if the allocation belongs to the core (see Proposition~\ref{prop1}). Moreover it \lq corrects\rq \ $DFS_i $ since it replaces the interval of possible values $ [0,c_i] $ for $ x_i $ with the interval $[min_i,max_i]$. Consequently, if for the DFS rate the maximum satisfaction for $i$ is measured when it gets $c_i$ and the minimum when it gets 0, with PS,  $i$ is measured to be totally satisfied when it gets $ max_i $ and totally unsatisfied when it gets $min_i$.

\begin{example}
	Consider $(c,E)$ of Example 1 (see Fig.\ref{es}) and the corresponding bankruptcy game model. It holds:\\
	Proportional allocation: $ DFS_2=0.555 $  and $ PS_2=0.444 $\\
	MMF allocation: $ DFS_2=0.3846$  and $ PS_2=0 $.\\
	In both cases the PS rate shows that player 2 is less satisfied than what expected with the DFS rate. This is due to the fact that the game guarantees player 2 to get at least 5. 	
\end{example}
The following propositions show some interesting properties of the PS rate.
\begin{proposition} \label{prop1}
	If the allocation $x$ belongs to the core of the bankruptcy game, $PS_i \in [0,1] \ \ \forall i \in N$.	
\end{proposition}
\begin{proof}
	If a solution $x$ belongs to a core it holds:
	$ x_i\geq v(i)$ and $x_i\leq v(N)-v(N\setminus i)$. Thus $v(i)$ and $ v(N)-v(N\setminus i) $ are the minimum and the maximum value that an allocation in the core can take.
	If $x_i=v(i)=min_i$ then $ PS_i=0 $, if $x_i=v(N)-v(N\setminus i)=max_i$ then $ PS_i=1$.
\end{proof}

\begin{proposition}\label{relaz}
	It is possible to summarize the bankruptcy regimes of the $PS$ rate in four possible cases as in Table~\ref{casi}.
	\begin{table}[h!]
    \vspace{-0.35cm}
		\begin{center}
			\begin{tabular} {c|c|r||c|r}
				\rule[-6pt]{0mm}{18pt}
				& \multicolumn{2}{c||}{$c_i< E$} & \multicolumn{2}{c}{$c_i\geq E$} \\
				& \textit{PS} & \textit{case} & \textit{PS} & \textit{case} \\
				\hline
				\rule[-6pt]{0mm}{18pt}
				$v(i)=0$ & $\frac{x_i}{c_i}$ & \textsc{Gm} &$\frac{x_i}{E}$ & \textsc{Gg}\\
				\rule[-6pt]{0mm}{18pt}
				$v(i)\neq0 $&$ \frac{x_i-v(i)}{c_i-v(i)}$& \textsc{Mm} & $\frac{x_i-v(i)}{E-v(i)}$ & \textsc{Mg}\\
			\end{tabular} 
		\end{center}
		\caption{Value of PS in the four possible cases.}
		\label{casi}
        \vspace{-0.3cm}
	\end{table}	
\end{proposition}
\begin{proof}
Let us treat each possible cases of Table~\ref{casi}:
\begin{itemize}
	\item \textit{Case \textsc{Gm}: $v(i)=0$, $c_i<E$}\\		
	Using the definition of bankruptcy game, it holds:\\ $v(N)-v(N\setminus i)= E-max\{0, E-c_i\}=E-E+c_i$. It follows $PS_i= x_i / c_i$.
	\item \noindent\textit{Case \textsc{Gg}: $v(i)=0$, $c_i\geq  E$}\\
	Using the definition of bankruptcy game, it holds:\\ $v(N)-v(N\setminus i)= E-max\{0, E-c_i\}=E$. It follows $PS_i= x_i / E$
	\item \noindent\textit{Case \textsc{Mm}: $v(i)\neq0$, $c_i<E$}\\
	As in case \textsc{Mg}, $v(N)-v(N\setminus i)= E-max\{0, E-c_i\}=c_i$. It follows $PS_i= (x_i-v(i)) / (c_i-v(i))$.
	\item \noindent\textit{Case \textsc{Mg}: $v(i)\neq0$, $c_i\geq E$}\\	
	As in case \textsc{Gg} , $v(N)-v(N\setminus i)= E-max\{0, E-c_i\}=E$. It follows $ PS_i= (x_i-v(i)) / (E-v(i))$. \qedhere
\end{itemize}
\end{proof}

\subsubsection*{Case terminology} the PS rate differentiates 4 possible cases we name \textsc{Gm}, \textsc{Gg}, \textsc{Mm}, \textsc{Mg}. If a player asks less than $ E $ we call it \textit{moderate player (\textsc{m})} while if it asks more than $ E $ it is a  \textit{greedy player (\textsc{g})}. In similar way, if the sum of the demand of a group of $ n-1 $ players exceeds $ E $, that means $ v(i)=0 $, the group is a \textit{group of greedy players (G)}
otherwise if $ v(i)\neq 0 $ we have a \textit{group of moderate players (M)}.

Proposition \ref{relaz} highlights that, not only there is a relation between the DFS rate and the PS rate, the satisfaction of a user should be modified when it is considered as a player inside a cooperative game. In particular, we can notice that for case \textsc{Gm} 
the PS rate coincides with the DFS one, i.e., $PS_i=DFS_i$; 
for case \textsc{Gg}, the user satisfaction measured with the PS rate is higher than when it is measured with the DFS rate, i.e., $PS_i \geq DFS_i$; 
in the \textsc{Mg} case, we have instead that  $DFS_i \geq PS_i$.
We can also notice that the denominator of the PS rate is always different from zero. In cases \textsc{Gm} and \textsc{Gg} this is obviously true, in case \textsc{Mm} the denominator is zero when $ \sum_{i=1}^n c_i=E $ but in this case we set $ PS_i=1 $ and in case \textsc{Mg} the denominator is zero when $ \sum_{j\in N,j\neq i} c_j=0 $ that is impossible. 
Furthermore, from Proposition \ref{relaz} it follows that if an allocation, i.e. a solution of an allocation problem that satisfies efficiency, non-negativity and demand boundedness, is an imputation, then $ PS_i \in [0,1] $ for all the users. This holds due to the fact that for an allocation, in each of the 4 cases presented above, it is always verified that  $v(N)-v(N\setminus i)$ is an upper bound for $x_i$. 


\subsection{Game-theoretical interpretation}
To support and justify the use of the new satisfaction rate, we show an interesting game-theoretic interpretation. 

Gately~\cite{Gately} introduced the concept of propensity to disrupt in order to eliminate the less fair imputation inside of the core. The idea was to investigate the gain of the player from the cooperation or, instead, its propensity to leave the cooperation, and to eliminate the imputation for which the propensity to leave the coalition for some players is excessively high. The formal definition of the propensity to disrupt is given in~\cite{little}.
\begin{definition}	[Propensity to disrupt]
	For any allocation vector $x$, the {\it{propensity to disrupt}} $d(x,S)$ of a coalition $S \in N$ $(S\neq \emptyset, N)$ is the ratio of the loss incurred by the complementary coalition $N\setminus S$ to the loss incurred by the coalition $S$ itself if the payoff vector is abandoned. In formula,
	\begin{equation}\label{d}
	{d(x,S)=\frac{x(N\setminus S)-v(N\setminus S)}{x(S)-v(S)}.}
	\end{equation}
\end{definition}
An equivalent definition of $d(x,S)$ is :	
	\begin{equation}\label{d2}
	d(x,S)=\frac{\widetilde{x}(S)-v(S)}{x(S)-v(S)}-1
	\end{equation}
	where:	$\widetilde{x}(S)= v(N)-v(N\setminus S)$ \cite{Gately}.

The propensity to disrupt of a coalition $S$ quantifies its desire to leave the coalition. When $x(S)=v(S)$ the propensity to disrupt of $S$ is infinite and the desire of $S$ to leave the coalition is maximum; when $x(S)>v(S)$ but $x(S)-v(S)$ is small, the value of $d(x,S) $ is very high and again $S$ does not like the agreement; when $x(S)=v(N)-v(N\setminus S)$ the propensity to disrupt is zero and $S$ has the propensity not to destroy the coalition; when $x(S)>v(N)-v(N\setminus S)$ the index is negative and there is an hyper-enthusiasm for such an agreement.

It holds an interesting relationship between the propensity to disrupt and the player satisfaction rate.
\begin{proposition}
	The relationship between the player satisfaction rate and the propensity to disrupt is:
	\begin{equation}
	\label{Is}{PS_i=\dfrac{1}{d(x,i)+1}.}	
	\end{equation}
	
\end{proposition}
\begin{proof}
	Using the alternative definition of $d(x,i) $ we have
	\begin{equation}
	{d(x,i)=\frac{v(N)-v(N\setminus i)-v(i)}{x_i-v(i)}-1}
	\end{equation}
	but $ \frac{v(N)-v(N\setminus i)-v(i)}{x_i-v(i)} $ is equal to $ \frac{1}{PS_i} $ so $d(x,i)=\frac{1}{PS_i}-1$.
\end{proof}
It is worth noting that if $d(x,i)$ goes to infinity, then $PS_i$ goes to $0$  and if $d(x,i)=0$ then $PS_i = 1$. This gives another interpretation of the PS rate. The higher the satisfaction  is, the bigger the enthusiasm of $i$, for being in the coalition, is. On the contrary, the closer to zero the user satisfaction is, the higher the propensity of user $i$ to leave the coalition is. 


\section{The Mood Value and the Player Fairness Index}
\label{moodvalue}


In this section, we define a new resource allocation rule we call the Mood Value. The fairness idea behind this rule is the same of the one behind the Jain's index. A repartition of a resource is fair when all the users have the same satisfaction.
Furthermore,  we propose a novel fairness index as a modification of the Jain's index.
\subsection{The Mood Value}
 Using the defined PS rate, 
 we can define the mood value. 
\begin{definition}[Mood Value]
	Given an allocation problem characterized by $(c,E)$, the allocation $x$ such that $PS_i=PS_j$ $\forall i,j \in N $ is called  \textit{mood value}.
\end{definition}
Due to the relation between the propensity to disrupt and the player satisfaction, the fairest solution corresponds to the one in which every player has the same propensity to leave the coalition. Equalizing the propensity to disrupt of the users, this allocation equalizes the mood of each player. In particular, given a game, it exists a unique mood such that the satisfaction of each user is the same. The closer to zero the mood is, the more unsatisfied user $i$ is; the closer to one the mood is, the more enthusiast the user $i$ is.

\begin{proposition}
	Let $(c,E)$ characterize an allocation problem. It exists a unique mood $m$ such that $PS_i=m$ $ \forall i \in N $; it~is:
	\begin{equation}\label{m}
	{	m=\frac{E-min }{max-min}}
	\end{equation}	
	where  $ min=\sum\limits_{i=1}^n v(i)= \sum\limits_{i=1}^n min(i)$ and $ max=\sum\limits_{i=1}^n [E-v(N\setminus i)]=\sum\limits_{i=1}^n max(i)  $.
	And the mood value is given by: 
	\begin{equation}\label{xm}
	{ x_i^m=v(i)+m(max(i)-min(i)). }
	\end{equation}
\end{proposition} 

\begin{proof}
	Let $PS_i=m $ $\forall i \in N$. It follows:
    
    $ x_i=m(E-v(N\setminus i))+(1-m)v(i).$\\
Due to the efficiency property it holds: 
		$\sum\limits_{i=1}^n m(E-v(N\setminus i))+(1-m)v(i)=E$. Thus \eqref{m}. 
	Since $x_i$ is the mood value iff $ PS_i=m$ $\forall i \in N $:
	\begin{equation}
	{\frac{x_i-v(i)}{E-v(N\setminus i)-v(i)}=m}
	\end{equation}
	$\forall i \in N$ and \eqref{xm} remains proved.
\end{proof}
\vspace{-0.2cm}

From \eqref{m} we can notice that the mood depends only on the game setting, thus, given a bankruptcy game, we can know a priori the value of the mood that produces a fair allocation. Knowing $m$, on can easily calculate  the mood value $ x_i^m $.

The formula \eqref{xm} shows that each user receives the minimum possible allocation $v(i) $ plus a portion $m$ of the quantity $max_i-min_i$. The nearer to 1 is the mood $m$, the greater is the happiness of each user, and the closer to the maximum the allocation is. In fact, when $m$ is equal to 1, the player receives exactly $E-v(N\setminus i)$, that is the maximum portion of resource that it can get, being inside a bankruptcy game.  

The mood value owns some interesting properties. It is an allocation thus it satisfies  non-negativity, demand boundedness and efficiency property; it is stable, that means it belongs to the core of the game (prop.~\ref{propcore}) and it guarantees more than minimal right to each player ($ x_i^m>v(i)) $.  Furthermore it satisfies the following property: if $ v(i)=v(j) $ and $ v(N\setminus i)=v(N\setminus j) $ then $ x_i^m=x_j^m $. This implies the equal treatments of equals ($ c_i=c_j $ $ \Rightarrow $ $ x_i^m=x_j^m $) and equal treatment of greedy claimants
(given a bankruptcy game, let G be the set of greedy players, i.e. such that $ c_i>E $: if $ |G|\geq 2 $ then $ x_i^m=x_j^m $ $ \forall i,j \in G $). This last property guarantees that even if a user has a cheating behavior its demand is bounded by the available amount of resource $E$. Furthermore the mood value is a strategy-proof allocation because a user has no advantages in splitting his demand.  
\begin{proposition}\label{propcore}
The mood value belongs to the core of $(N,v)$.
\end{proposition}
\begin{proof} 
We should prove that $x_S^m\geq v(S)$, $\forall S \subseteq N$.\\
If $ v(S)=0 $ the condition holds due to the fact that $ x_i^m< 0 $, $ \forall i \in N $. 
Now consider the case $ v(S)>0$. Suppose that $x_S^m < v(S)=E-\sum_{i \in N \setminus S} c_i$. For the efficiency property it holds  $ E=x_S^m+x_{N\setminus S}^m $,
implying  $ x_{N\setminus S}^m>\sum_{i \in N \setminus S} c_i$, which yields a contradiction with the fact that, according to the mood value solution, each user receives at most its demand.
\end{proof}

 \subsubsection*{Mood Value Computation Complexity} 
 Differently from the other allocation solutions inspired by game theory, in order to calculate this new allocation, only the value of $ 2n $ coalitions, i.e., the ones formed by the single players and the ones containing $ n-1 $ players, is needed. 
The time complexity of mood value computation is dominated by the complexity of computing $v(i)$ that is $\mathcal O(n)$. In dynamic situations, i.e. when the value of each of the $n$ coalitions has to be updated at each slot of time, the complexity is therefore $\mathcal O(n^2)$, but it can be reduced to $\mathcal O(n)$ where $v(i)$ pre-computation is possible. This makes the mood value the best allocation rule in terms of time complexity together with the proportional allocation: the Shapley value has a time complexity of $\mathcal O(n!)$, while iterative algorithms for the computation of MMF and CEL allocations have a $\mathcal O(n^2\log n)$ time complexity; the Nucleolus computation is a NP-hard problem.

In terms of spatial complexity, the mood value, proportional, MMF and CEL allocations can be considered as equivalent and in the order of $\mathcal O(n)$. Instead, the Shapley value and the Nucleolus computations have a spatial complexity of $\mathcal O(2^n)$. 

\subsection{The Player Fairness Index}
Considering the observed good properties that make the Jain's index a strong fairness index, we propose its modification replacing the DFS rate of the Jain's index with the PS rate.  The resulting new fairness index we propose takes value 1 when all the users have the same satisfaction, i.e., when the allocation is the mood value.
\begin{definition}[Players fairness index]
	Given a problem $(c,E) $ and an allocation $x$, the \textit{players fairness index} is: 
	\begin{equation}
	\label{jainplayer}{J_p=\biggl[\sum\limits_{i=1}^n \bigl(PS_i\bigr) \biggr]^2 \bigg/ n\sum\limits_{i=1}^n \bigl(PS_i\bigr)^2 }
	\end{equation}
	
\end{definition}
\begin{proposition}
	The players fairness index takes value in $[\frac{1}{n}, 1]$ when the allocation belongs to the core.
\end{proposition}
\begin{proof}
	From Proposition \ref{prop1} follows that $ PS_i $ belongs to $ [0,1] $ and that  $ \sum\limits_{i=1}^n PS_i$ is always not negative. The maximum fairness is measured when all the users have the same PS rate, i.e.: 
	\begin{center}
		$\biggl[\sum\limits_{i=1}^n \bigl(PS_i\bigr) \biggr]^2 =\bigl( n PS_i \bigr)^2 \Rightarrow n\sum\limits_{i=1}^n \bigl(PS_i\bigr)^2=nn \bigl(PS_i\bigr)^2.$
	\end{center}
	Thus $ J_p=1 $.
	The minimum fairness is measured when $\exists ! k$ s.t. $PS_k\neq 0 $ and $ PS_j= 0 $ $\forall j \neq k$. In this case:
	
	{\small
	\hspace{-0.4cm}$\biggl[\sum\limits_{i=1}^n \bigl(PS_i\bigr) \biggr]^2 =\bigl(PS_k\bigr)^2 \Rightarrow n\sum\limits_{i=1}^n \bigl(PS_i\bigr)^2=n \bigl(PS_k\bigr)^2 \Rightarrow J_p=\frac{1}{n}$
	}
\end{proof}

\vspace{-0.2cm}
For core allocations, $J_p$ takes value in the same interval of $J$ making possible a comparison between the two indices. Furthermore, this index maintains all the good properties of the Jain's index: the population size independence, the scale and metric independence, the boundedness and the continuity.


\section{Interpretation with respect to traffic theory}
\label{traffictheory}

In the already cited seminal works about the definition of proportional and weighted proportional allocations in network communications, network optimization models are defined using as goal the maximization of an utility function. A typical application is the bandwidth sharing between elastic applications~\cite{prop}, i.e., protocols able to adapt the transmission rate upon detection of packet loss.  In this context we show how it is possible to revisit the mood value as a value resulting of the sum of the minimum allocation and the result of a weighted proportional allocation formulation where the weights are not the original demands, but new demands re-scaled accordingly to the maximum possible allocation knowing the available resource, and the minimum allocation under the awareness of other user's demands and the available resource. More precisely, the mood value can be computed as the result of the following 4-step algorithm.

\begin{flushleft}
\underline{\textit{Step 1}}:
We assign to each user the minimal right $ v(i) $.
\underline{\textit{Step 2}}: 
We set the new value of the estate $E'=E-min=$\\
$E-\sum\limits_{i=1}^{n}v(i) $ and the new demands $ c'_i = max_i-min_i$.\\
\ \\
\underline{\textit{Step 3}}:
We solve the following optimization problem 
\end{flushleft}
\begin{equation}
\begin{aligned}
& \underset{x}{\text{maximize}}
& & \sum\limits_{i=1}^{n}  c'_i\log x_i \\
& \text{subject to}
& & x_i \leq c'_i, \; i = 1, \ldots, n\\
& & &x_i  \geq 0, \; i = 1, \ldots, n\\
& & &\sum\limits_{i=1}^{n}  x_i=E'
\end{aligned}
\end{equation}
\ \\
\underline{\textit{Step 4}}:
The mood value coincides with the sum of the minimal right and the allocation given by step 3: $ x_i^m=v(i)+x_i $. 
\begin{proof}
 We should prove that the result of the optimization problem is $ x_i=mc_i'$. The lagrangian of the problem is \\
 $L(x,\mu, \lambda) = \sum\limits_{i=1}^{n}  c'_i\log x_i-\mu^T(C-Ax)-\lambda( E'-\sum\limits_{i=1}^{n}  x_i)$\\
where the vector $ \mu $ and $ \lambda $ are the lagrangian multipliers (or shadow prices), $ C $ is the vector of the demands $ [c_1', ...c_n'] $ and $ A $ is the identity matrix of dimension $ n $. Then, 
$  \frac{\partial L}{\partial y_i} =\frac{c'_i}{y_i}-\mu_i-\lambda$. The optimum is given by $ y_i=\frac{c'_i}{\mu_i+\lambda} $ when $ \mu\geq 0 $, $ Ay\leq C $, $  \sum\limits_{i=1}^{n}  y_i=E'$ and $ \mu^T(C-Ay)=0 $. This coincides with the case in which $ \mu^T=0 $ and $ \lambda\neq 0 $.
In fact, we have $ \sum\limits_{i=1}^{n}\frac{c'_i}{\lambda}=\frac{1}{\lambda}\sum\limits_{i=1}^{n}c'_i=E' .$ It follows that $ \lambda=\frac{1}{E'}\sum\limits_{i=1}^{n}c'_i $ is greater or equal to 1 and $ y_i=\frac{c'_i}{\lambda} $ is less or equal to $  c_i' $, that is an admissible solution. We can now notice that $ \lambda=\frac{1}{E'}\sum\limits_{i=1}^{n}c'_i=\frac{max-min}{E-min}=\frac{1}{m}$. It follows $ y_i=mc_i' $.
\end{proof}
\begin{example}
Let $(c,E)$ be the allocation problem of Fig. \ref{es}. Following the algorithm we have:\\
\textit{Step 1:} 
$ v(i)=[0,5,0] $. \hfill
\textit{Step 2:} 
 $E'=5$, $ c'_i = [3,5,2]$.\\
\textit{Step 3:}
$ x=[1.5, 2.5, 1] $ \hfill
\textit{Step 4:}
$ x_i^m=[1.5,7.5,1] $. \ \ \ \ \ 
\end{example}
The algorithm shows that the mood value firstly assign the minimal right (step 1) and secondly, considering the new allocation problem resulting after the first assignment (step 2), it allocates in a proportional way the resources (step 3). The proportion of resource to allocated  is the mood. 

We provides two ways to compute the mood value: (\ref{xm}) and the 4-step algorithm of Section~\ref{traffictheory}. It is clear that the computation of the mood value throught the formula (\ref{xm}) is less complex than the one using the 4 steps algorithm. 
%
%

\section{Numerical examples}
\label{results}

We tested the mood value and the new fairness index in a few significant configurations comparing them with the classical allocations and the Jain's index. 

We considered two demands distributions: (i) a uniform distribution, and (ii) a Weibull distribution.
The former can be considered as a baseline, while the latter a maybe more realistic one.
Taking inspiration from cellular (OFDMA) resource allocation studies we emulated an indoor scenario of femtocells using the WINNER II channel model~\cite{winner}:  generating in a uniform way 10000 users around the cell station between 3 and 100 m, we associate resource blocks (RBs) to each of them with a transmit power between 1 and 100 dB; 
the resulting RB distribution is well fit by a Weibull distribution.
We now consider a range for demand generation between 0 and 100 units and we generate from (i) a uniform distribution between 0 and 100 and (ii) from a Weibull distribution 
$f(x) = (\frac{a}{b})(\frac{x}{b})^{(a-1)} e^{- (\frac{x}{b})^a}$ for $x > 0$ with scale parameter $a=40$ and shape parameter $b=1.4$. 

It is worth noting that the Weibull distribution is quite close to the Pareto distribution (both are exponential ones), its and discrete variations (e.g., Zipf's one), for example used in in-network content caching resource allocation~\cite{bank} .
\begin{figure}[b]
	\vspace{-0.7cm}
	\subfloat[3 users, uniform]{\includegraphics[scale=0.14]{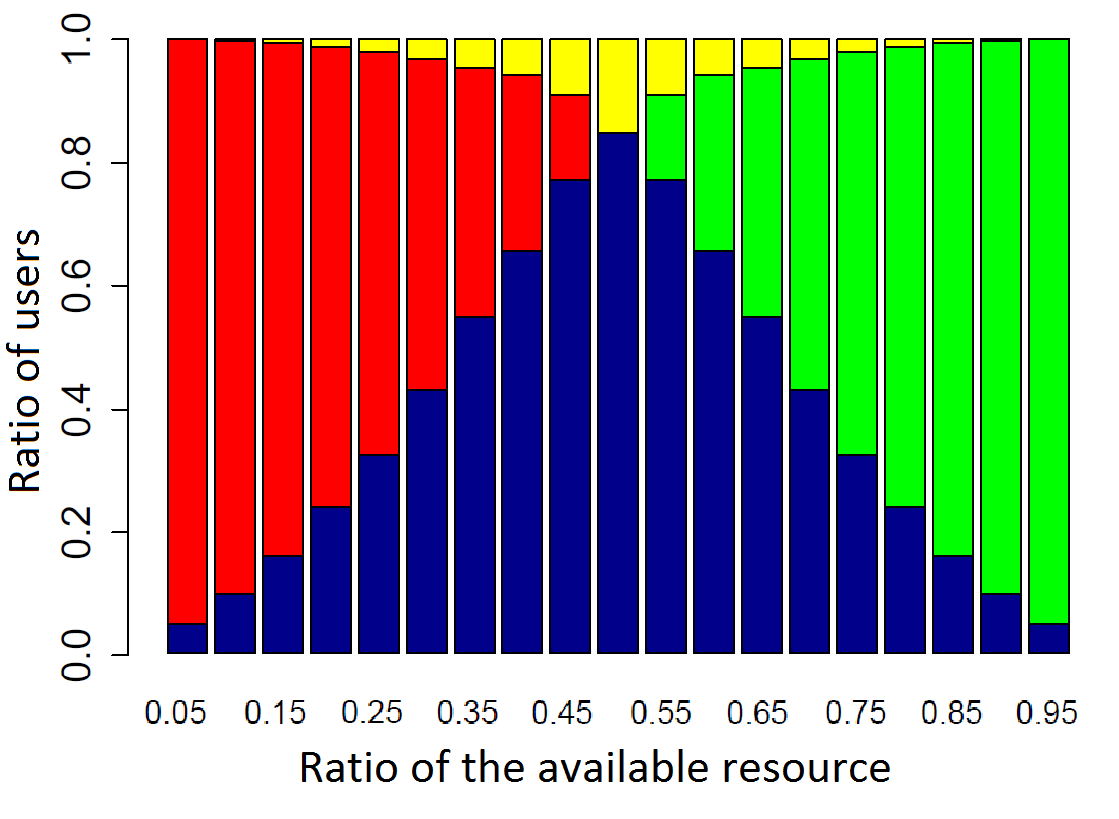}}%
        \hspace{-0.1cm}
	\subfloat[5 users, uniform]{\includegraphics[scale=0.14]{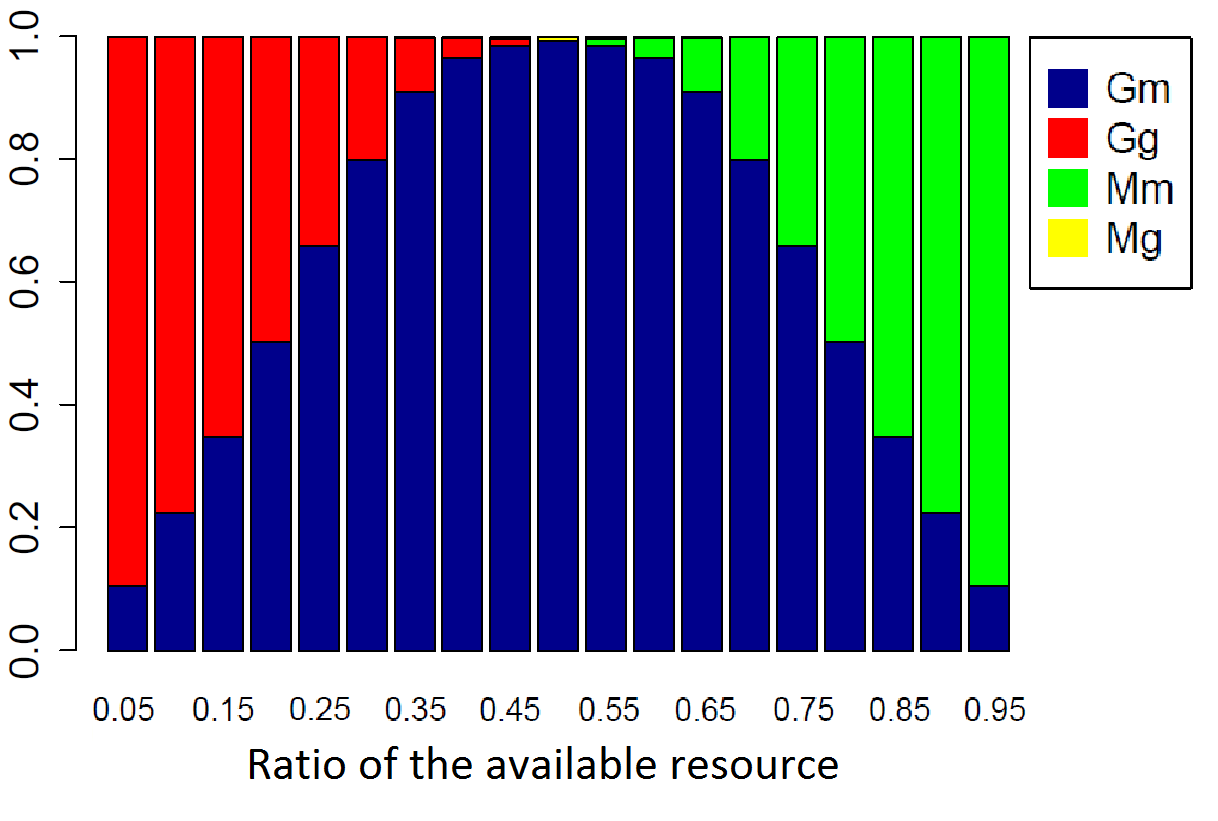}}%
     
     \vspace{-0.4cm}
	\subfloat[3 users, Weibull]{\includegraphics[scale=0.14]{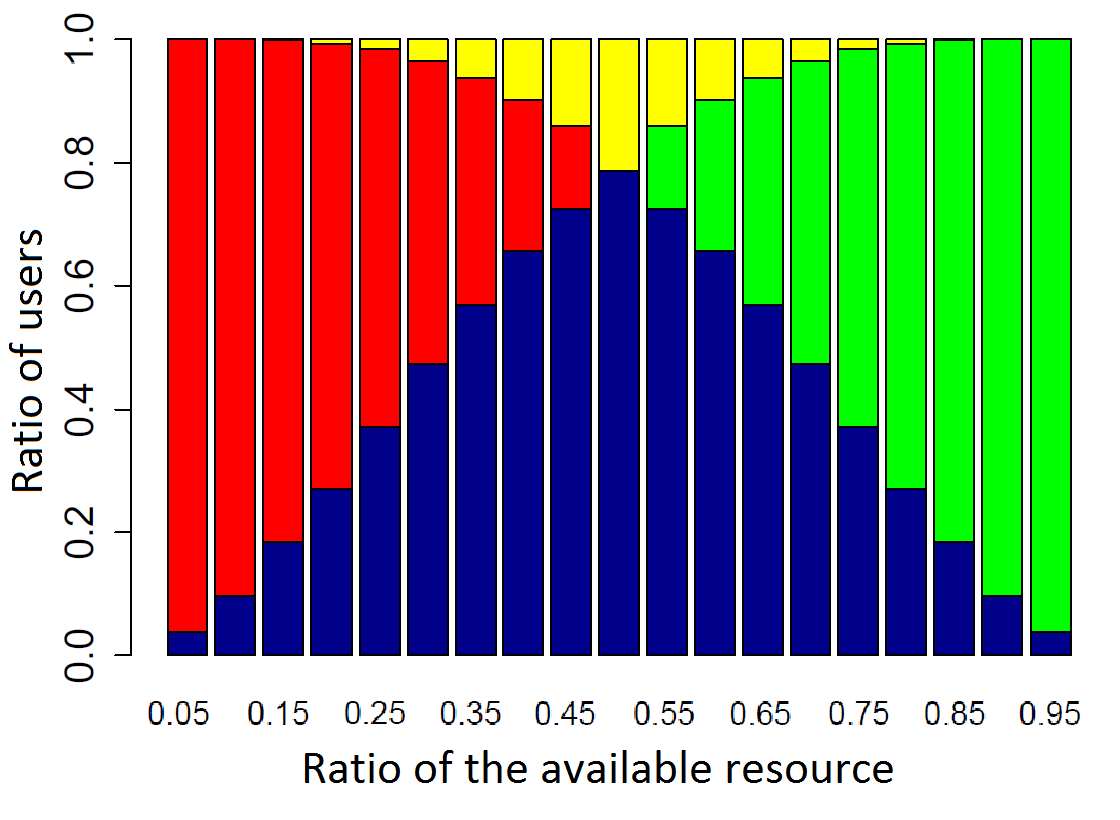}}%
	\hspace{-0.1cm}
	\subfloat[5 users, Weibull]{\includegraphics[scale=0.14]{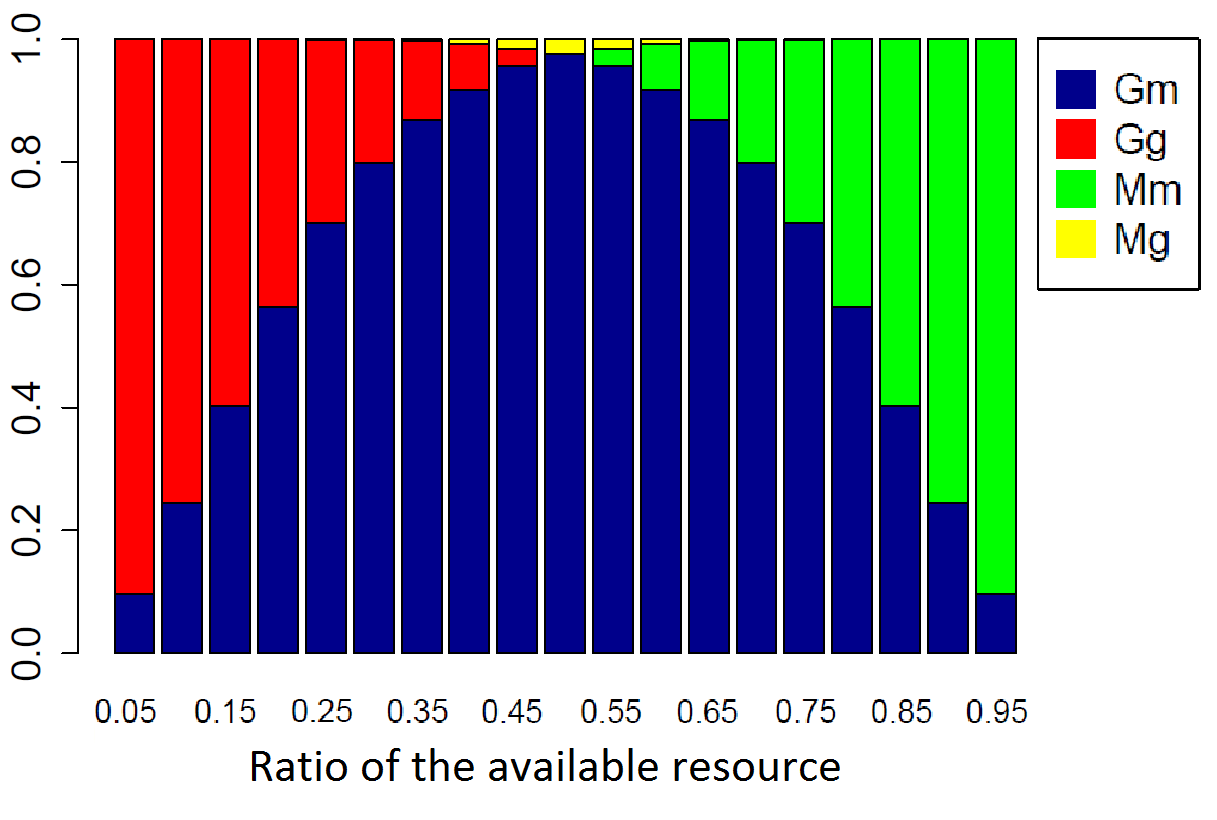}}%
	\caption{User cases distribution }
         \vspace{-0.5cm} 
	\label{scenari}
\end{figure}

\begin{figure}[t]
	\vspace{-0.6cm}
	\subfloat[Proportional]{\includegraphics[height=2.58cm]{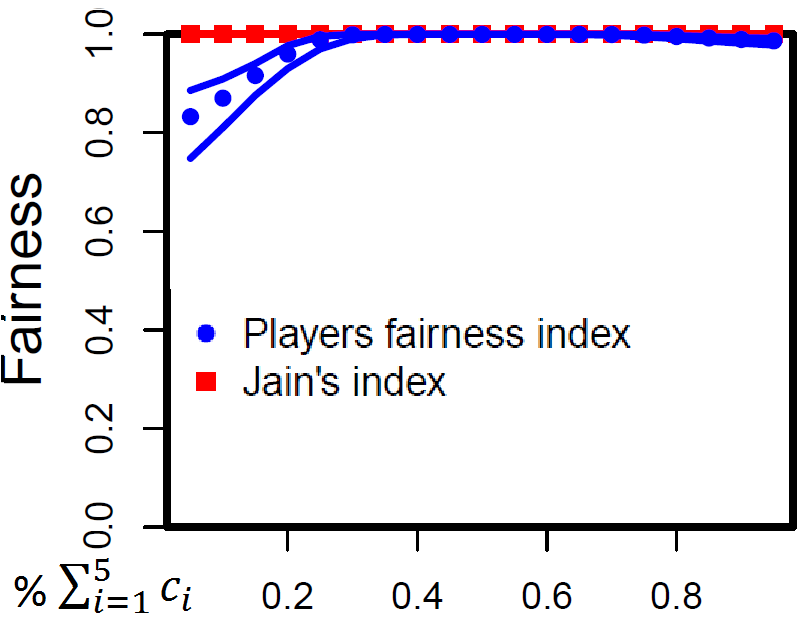}%
		\label{propr}}
	\hspace{-0.1cm}
	\subfloat[Shapley Value]{\includegraphics[height=2.58cm]{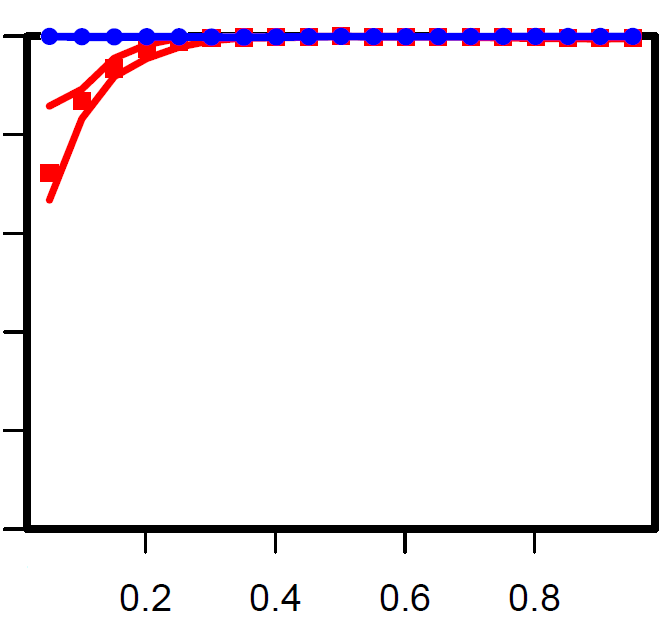}%
		\label{shapleyr}}
	\hspace{-0.1cm}
	\subfloat[Nucleolus]{\includegraphics[height=2.58cm]{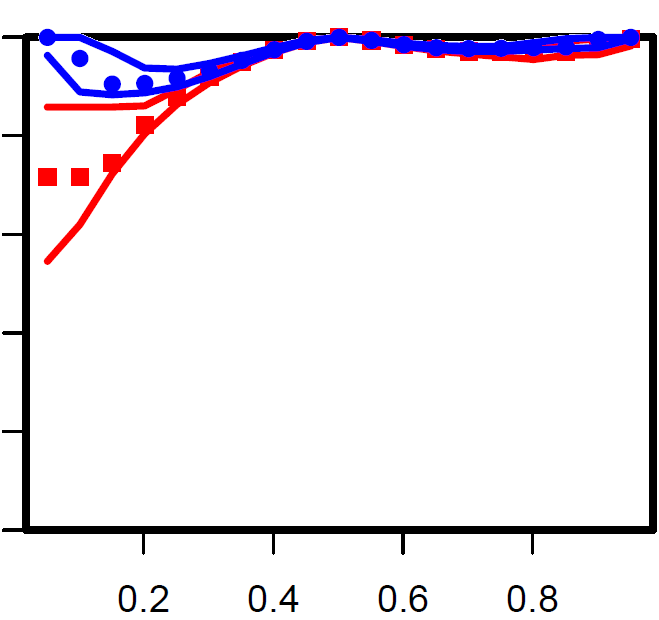}%
		\label{nuclr}}
	\vspace{-0.05cm}	
	\subfloat[Mood Value]{\includegraphics[height=2.58cm]{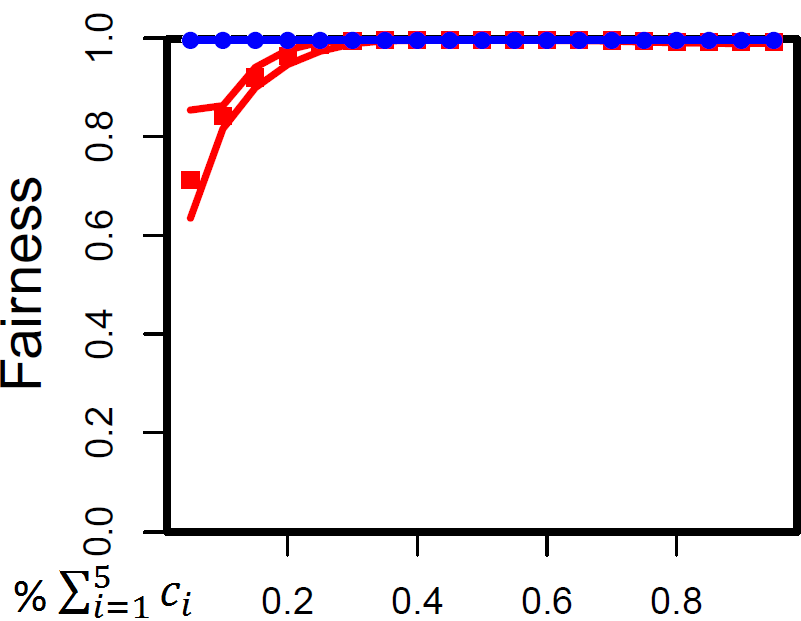}%
		\label{moodr}}
	\hspace{-0.1cm}
	\subfloat[MMF]{\includegraphics[height=2.58cm]{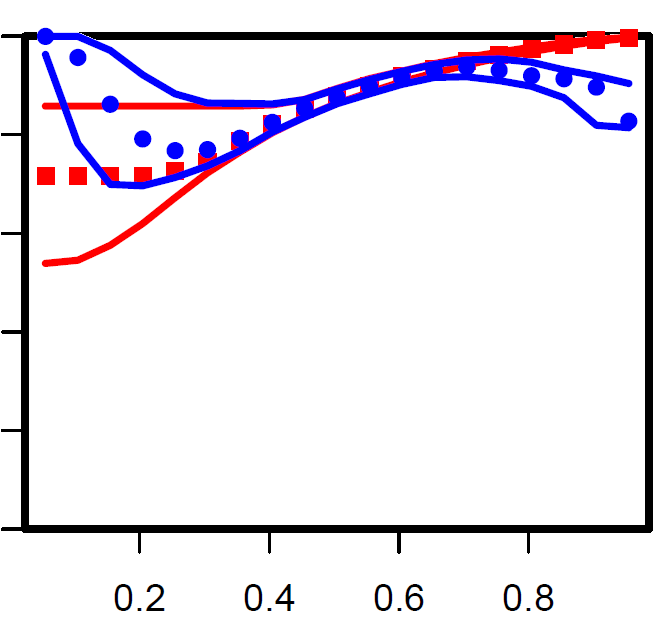}%
		\label{mmfr}}
	\hspace{-0.1cm}
	\subfloat[CEL]{\includegraphics[height=2.58cm]{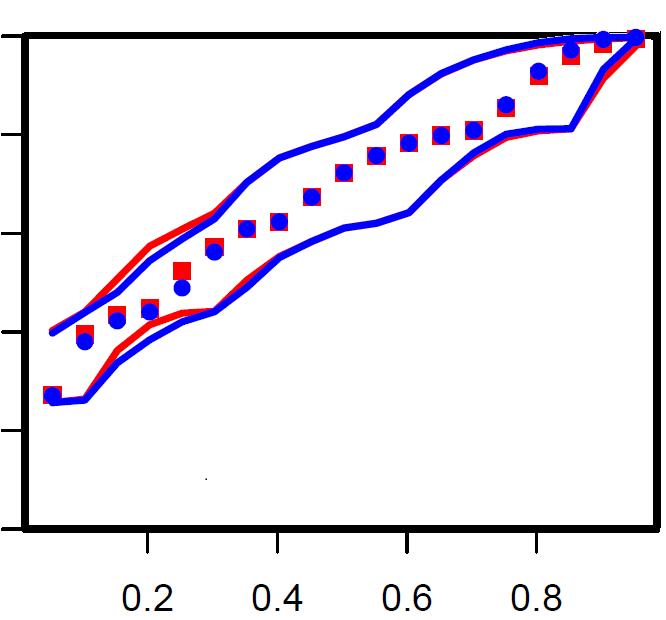}%
		\label{celr}} 	
	\caption{Fairness as a function of $E$/demand (5 users, uniform)}
        \vspace{-0.5cm}  
	\label{fig_simrand}
\end{figure}

We run different instances with a ratio of $ E  $ (available resource) ranging from $ 5\% $ to $ 95\%$ of the global demand.
We first simulate 300 bankruptcy games with 3 and 5 users.
Fig.~\ref{scenari} show the users configuration as a function of the available resource. 
With 3 users (Fig.~\ref{scenari}a,c), for low value of $ E $ almost all are greedy players (\textsc{Gg} case) due to the fact that the resource is small; increasing $E$ the number of moderate players (\textsc{Gm}) increases but also some users in configuration \textsc{Mg}  appear. In fact, increasing $ E  $ some greedy players become moderate while the others remain greedy ones; some of them are greedy inside a group of greedy users (\textsc{Gg}), while some others greedy inside a group of moderate ones (\textsc{Mg}).
When the available resource is higher than half of the global demand, greedy players \textsc{Gg} disappear and the number of moderate players increases. In particular, users \textsc{Mm} appear and they become the majority when the resource is large. 
With 5 users (Fig.~\ref{scenari}b,d), we find a similar trend than with 3 users in the number of moderate players that increases when increasing $ E $. However, \textsc{Mg} users are few; in fact, it holds that it can exist at most one \textsc{Mg} user in a game and, due to the higher number of users in the system, it is very unlikely that there exists only a player \textsc{Mg} in the system such that the sum of the demands of the other $ n-1 $ players exceeds $ E$. Thus, with a number of users higher than 5, one can practically reduce the number of user cases from 4 to 3. For this reason, in order to capture all the possible scenarios, we choose a low number of user for the first round of simulations.

Fig.~\ref{fig_simrand},~\ref{fig_sim1} and~\ref{fig_sim2} show the results of the first simulations.   We consider the six allocations discussed before: Proportional, Shapley, Nucleolus, Mood Value, MMF and CEL. We calculate the Jain's fairness index and the players fairness index and we plot, for each ratio of $E$ and each index, the mean value in between the first and third quantile lines. 
The Jain's index is depicted with the red color and square points while the players fairness index with the blue one and round points. Due to space limit, we do not show the case of 3 users with uniform demand distribution because similar to the 3-user one with the weibull distribution.

\begin{figure}[t]
	\subfloat[Proportional]{\includegraphics[height=2.58cm]{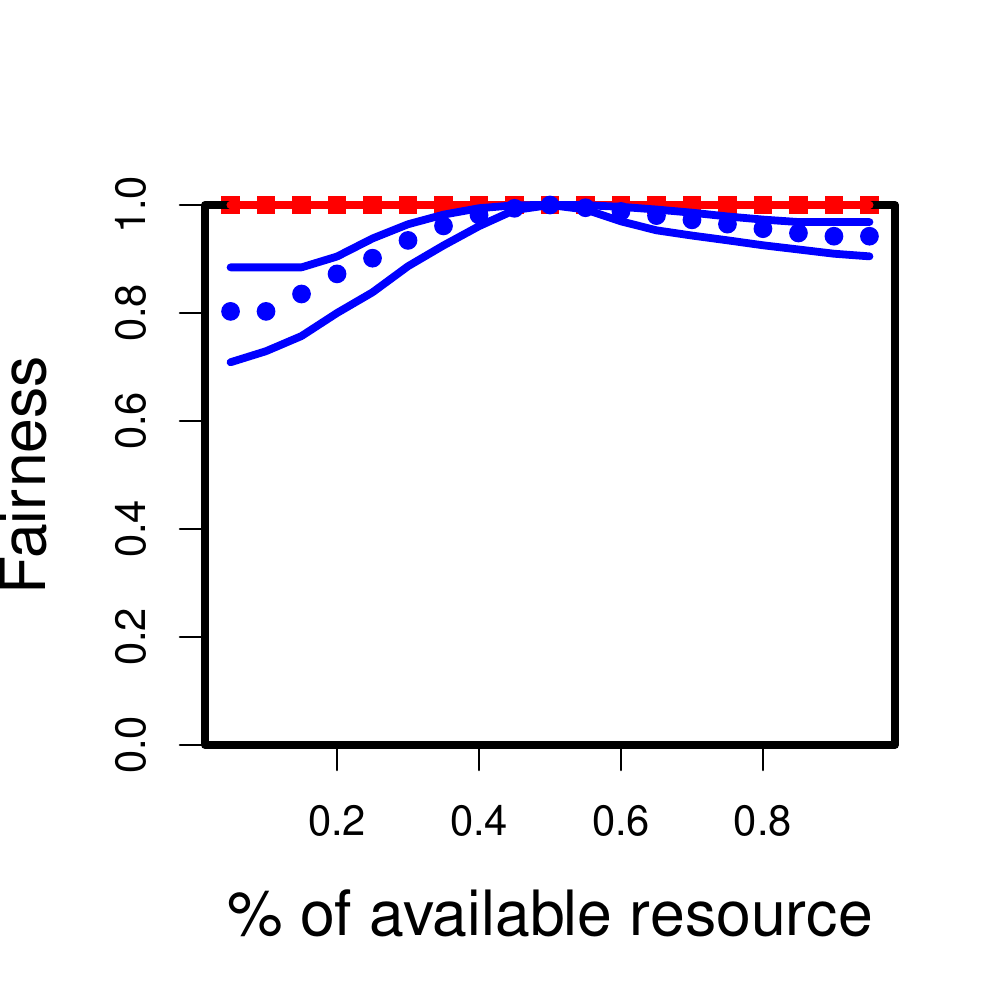}%
		\label{prop3}}
	\hspace{-0.1cm}
	\subfloat[Shapley Value]{\includegraphics[height=2.58cm]{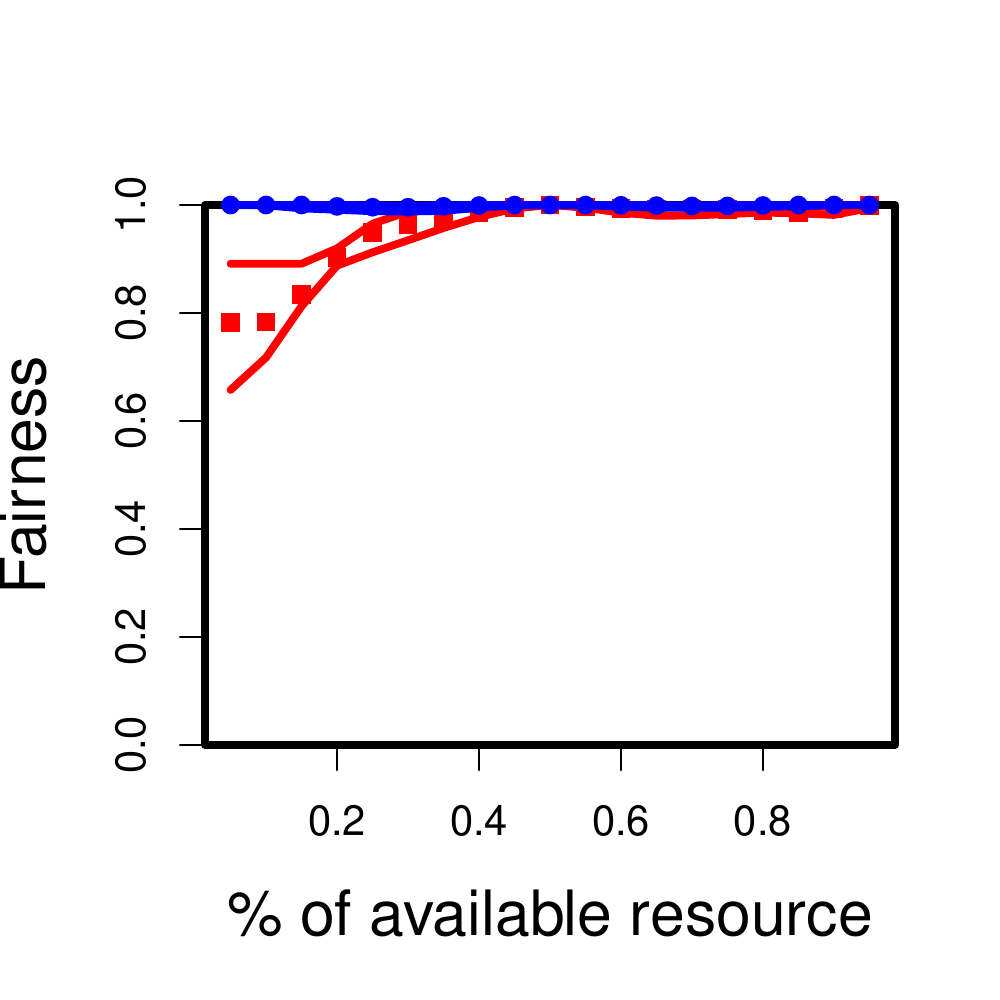}%
		\label{shapley3}}
	\hspace{-0.1cm}
	\subfloat[Nucleolus]{\includegraphics[height=2.58cm]{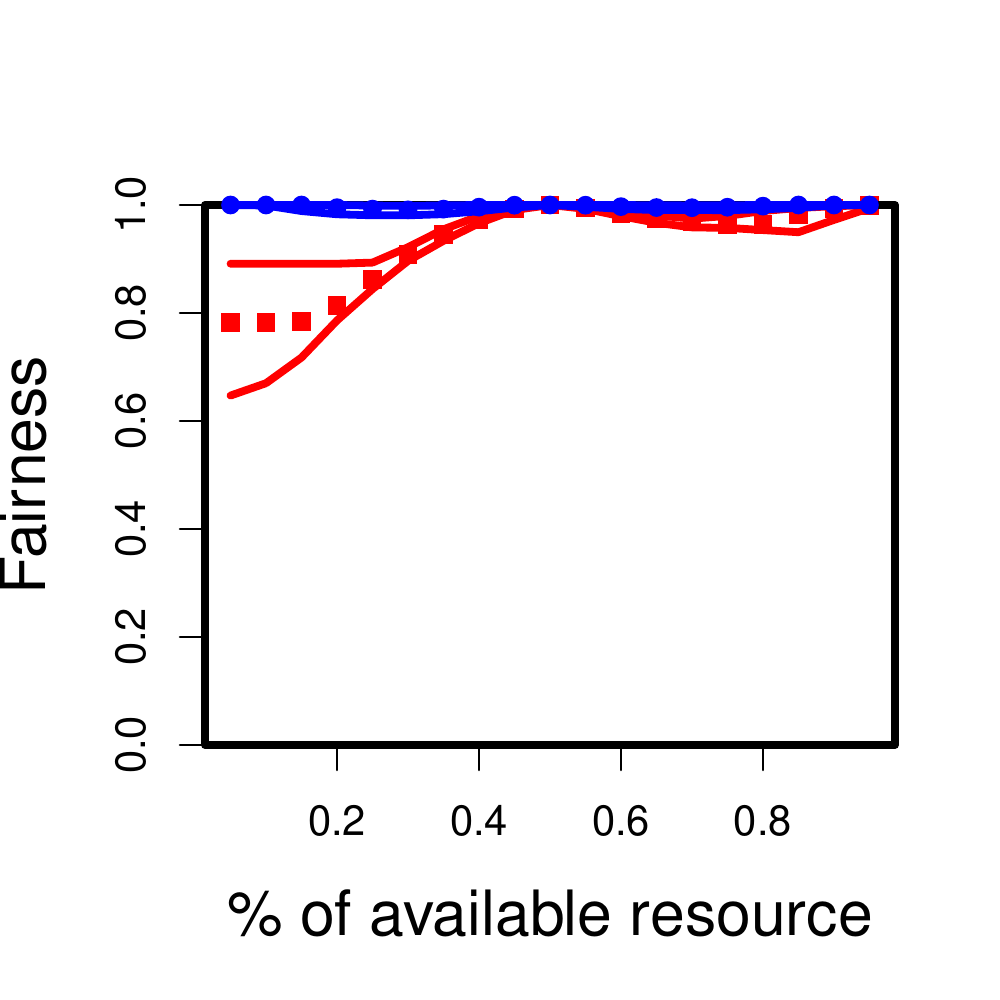}%
		\label{nucl3}}
	\vspace{-0.05cm}	
	\subfloat[Mood Value]{\includegraphics[height=2.58cm]{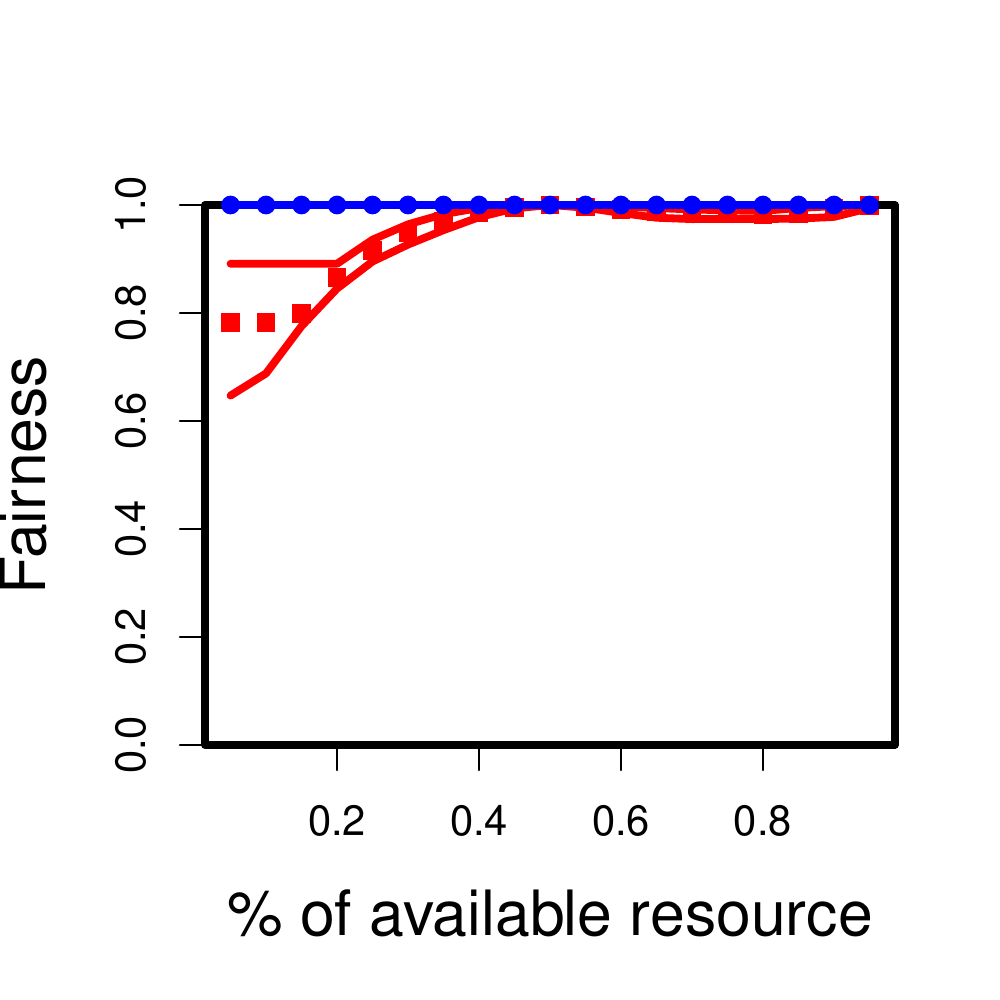}%
		\label{mood3}}
	\hspace{-0.1cm}
	\subfloat[MMF]{\includegraphics[height=2.58cm]{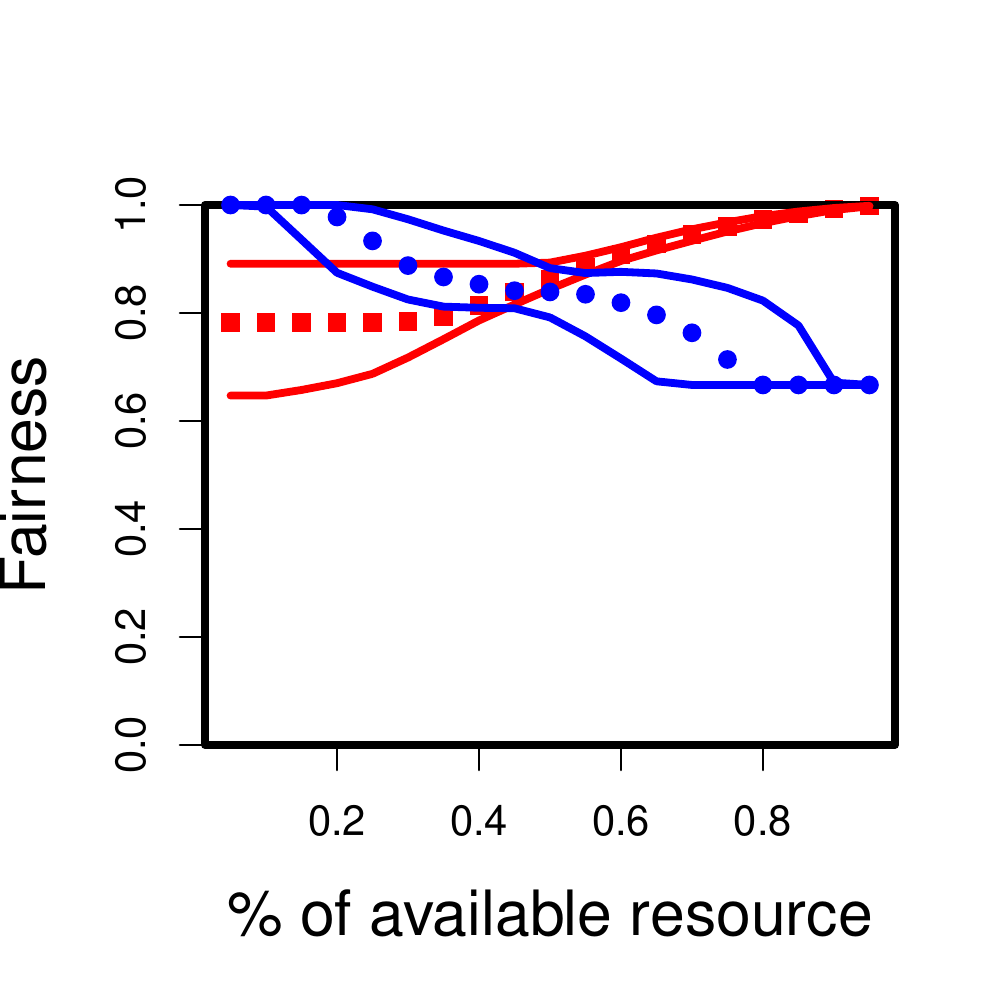}%
		\label{mmf3}}
	\hspace{-0.1cm}
	\subfloat[CEL]{\includegraphics[height=2.58cm]{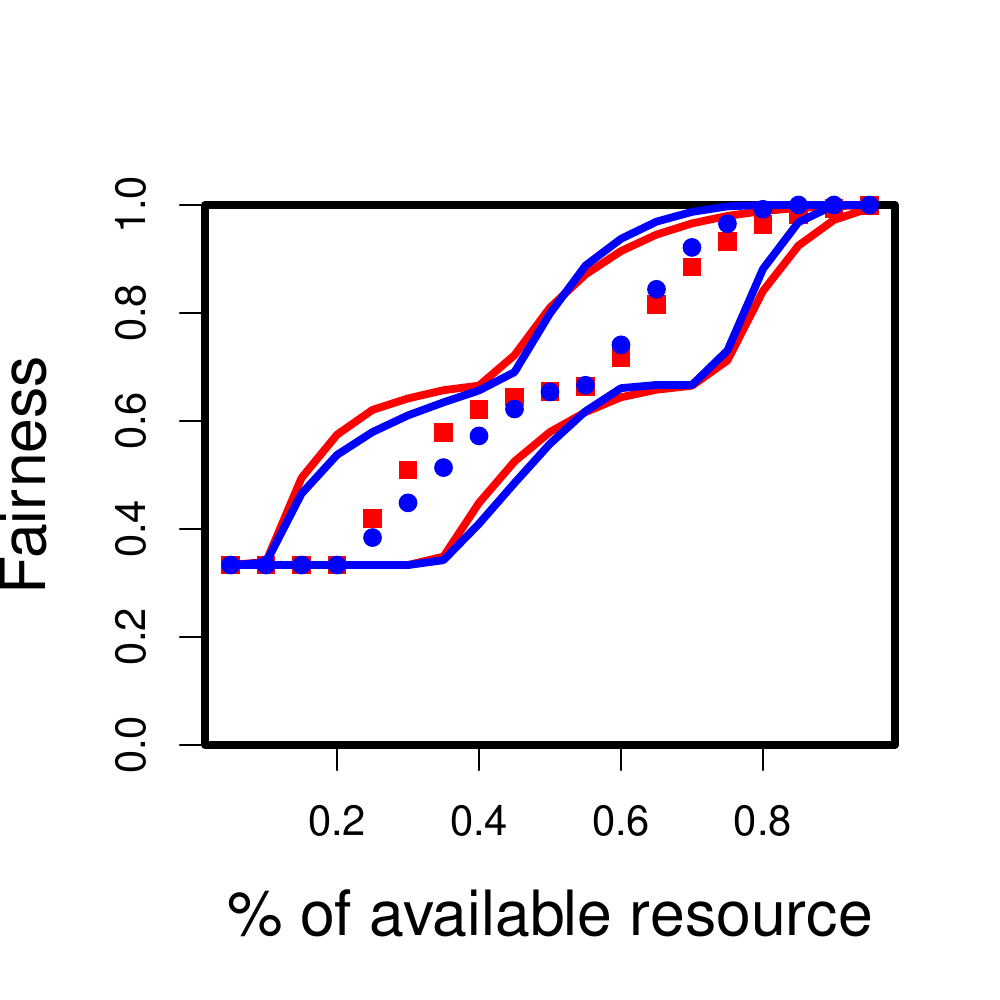}%
		\label{cel3}} 	
	\caption{Fairness as a function of $E$/demand (3 users, weibull)}
    \vspace{-0.5cm}  
	\label{fig_sim1}
\end{figure}

\begin{figure}[t]
	\vspace{-0.6cm}
	\subfloat[Proportional]{\includegraphics[height=2.58cm]{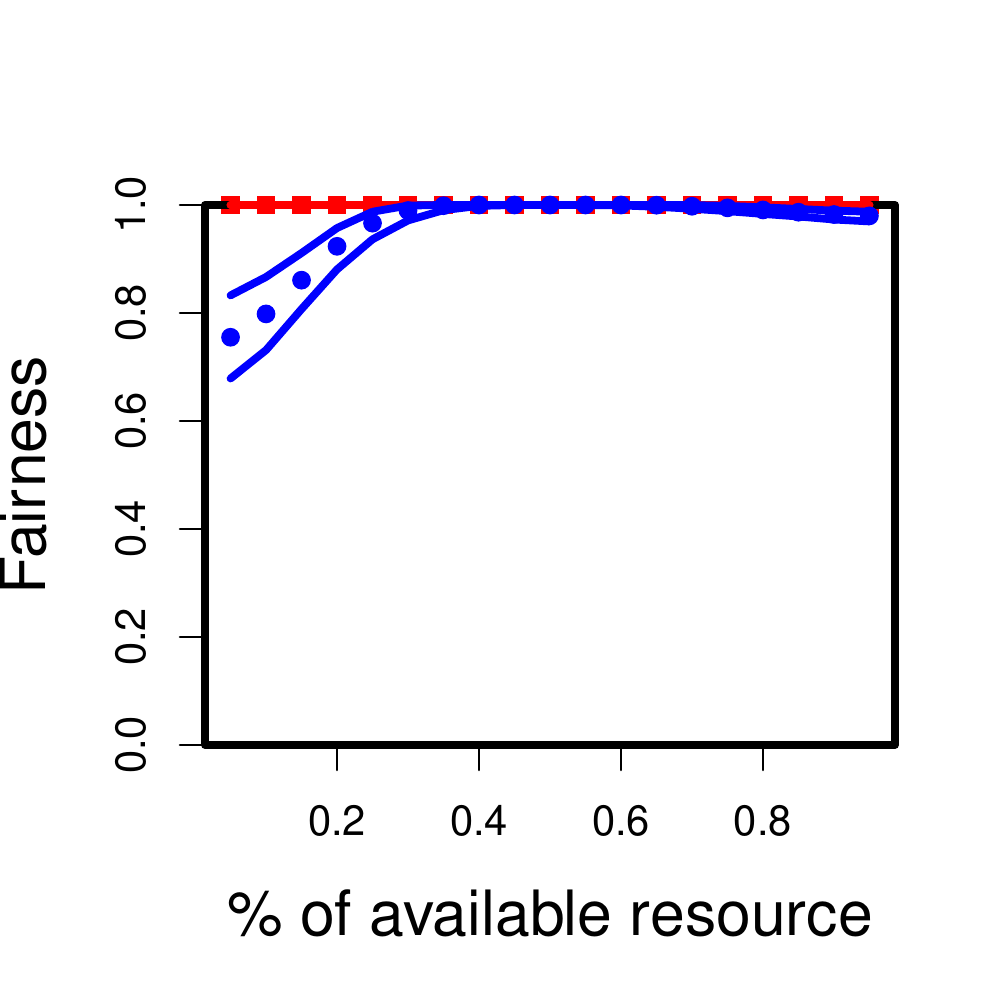}%
		\label{prop5}}
	\hspace{-0.1cm}
	\subfloat[Shapley Value]{\includegraphics[height=2.58cm]{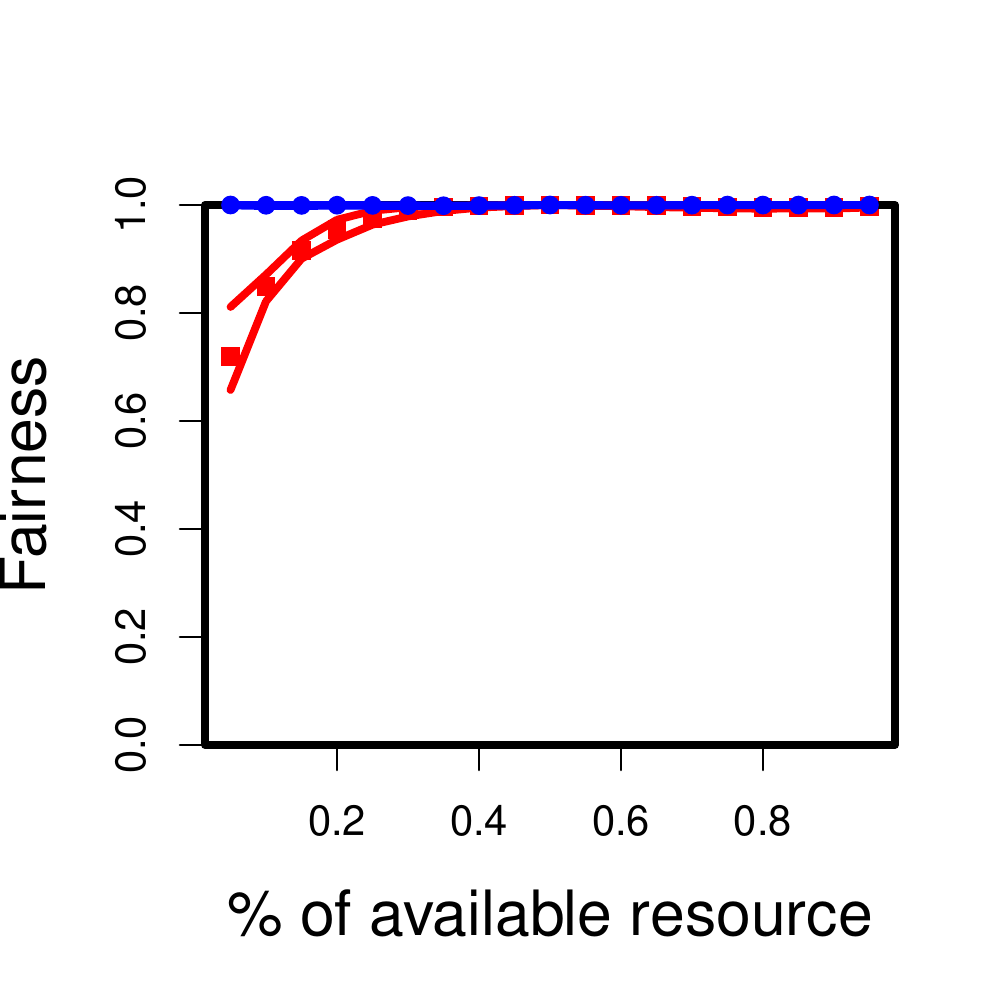}%
		\label{shapley5}}
	\hspace{-0.1cm}
	\subfloat[Nucleolus]{\includegraphics[height=2.58cm]{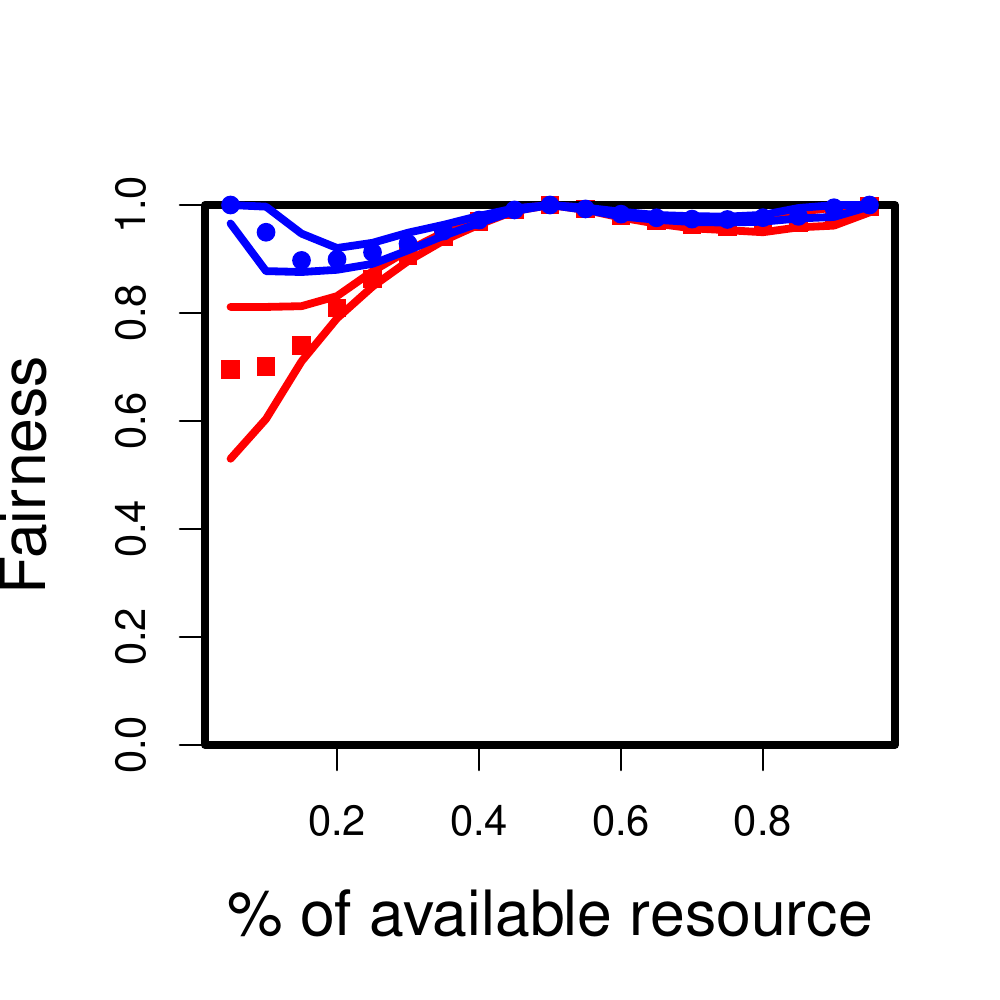}%
		\label{nucl5}}
	\vspace{-0.05cm}	
	\subfloat[Mood Value]{\includegraphics[height=2.58cm]{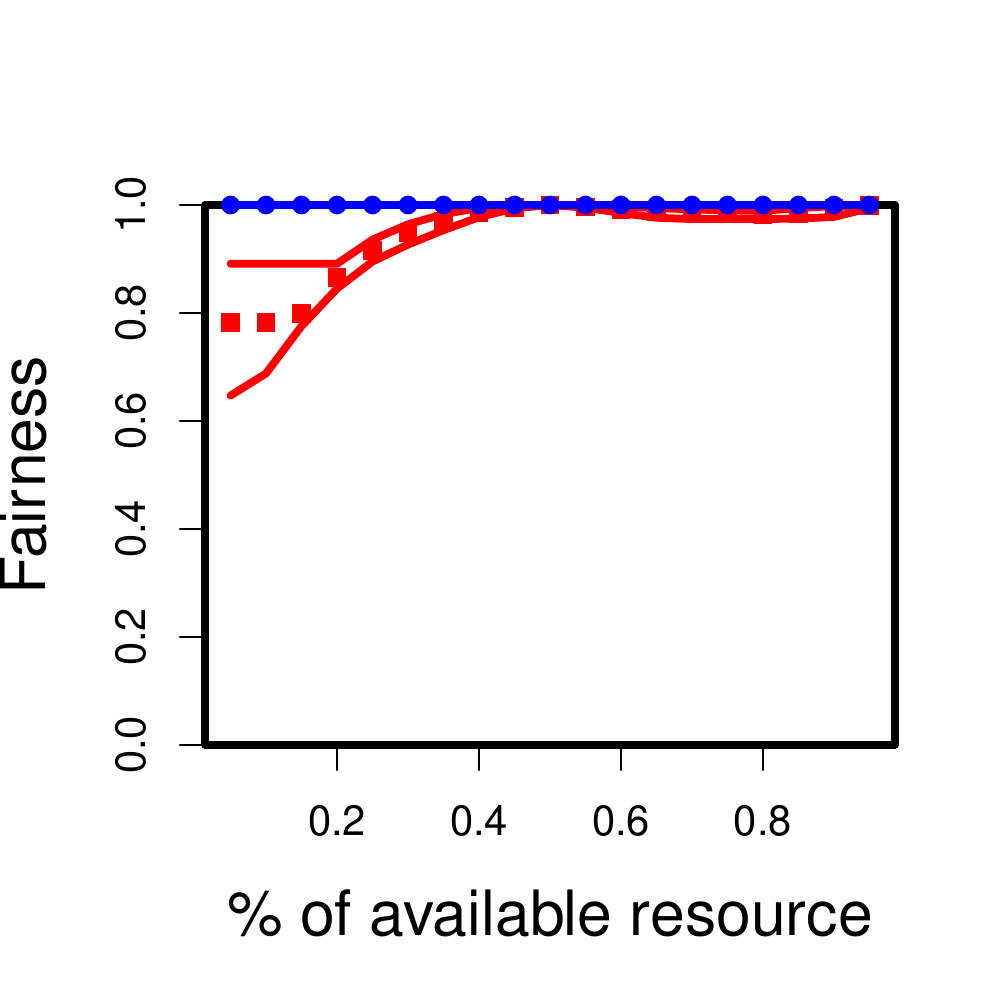}%
		\label{mood5}}
	\hspace{-0.1cm}
	\subfloat[MMF]{\includegraphics[height=2.58cm]{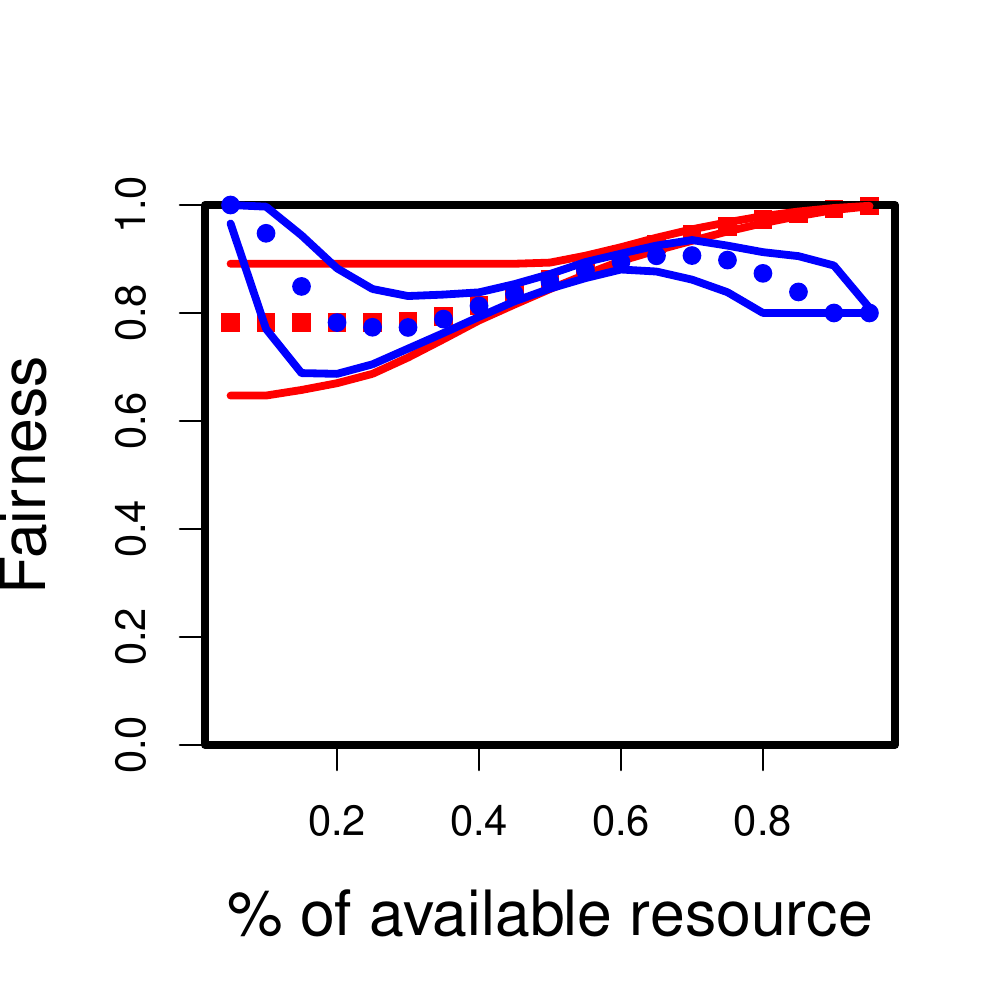}%
		\label{mmf5}}
	\hspace{-0.1cm}
	\subfloat[CEL]{\includegraphics[height=2.58cm]{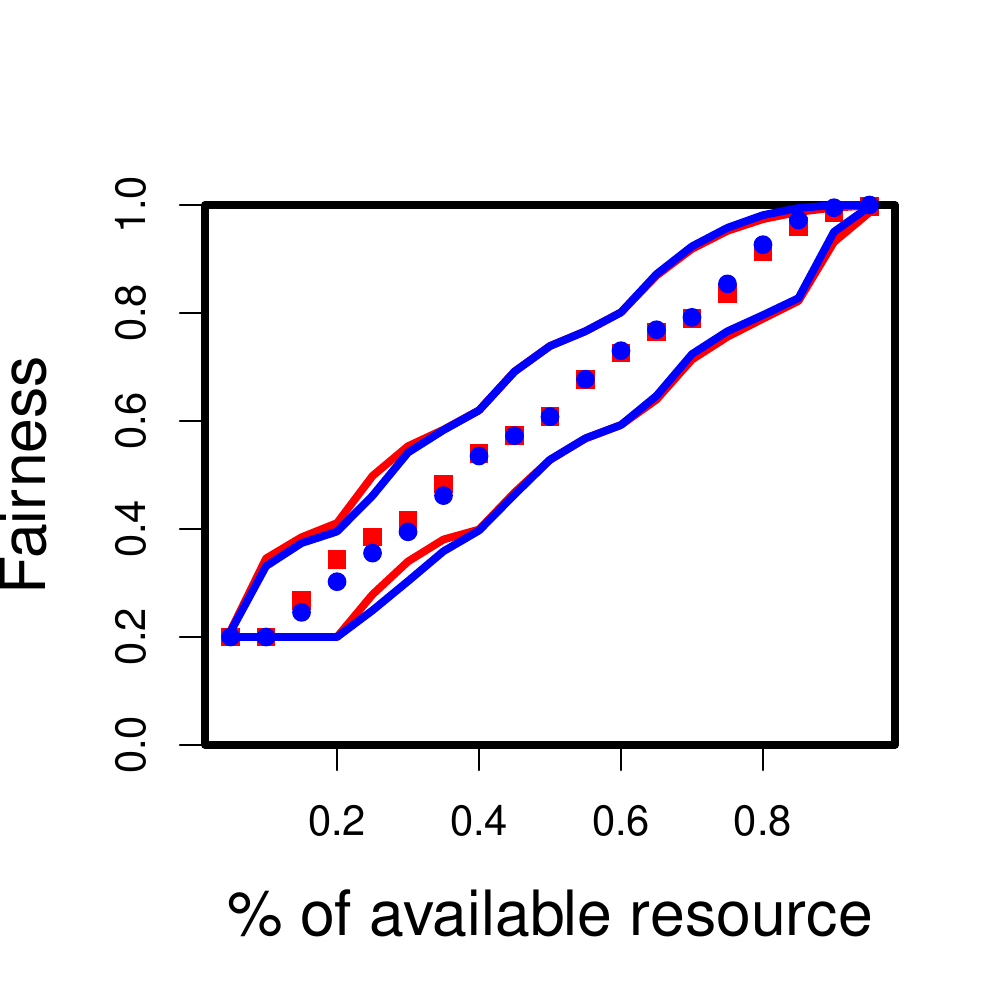}%
		\label{cel5}}	
	\caption{Fairness as a function of $E$/demand (5 users, weibull)}
    \vspace{-0.5cm}
	\label{fig_sim2}
\end{figure}

In the 3-user scenario (Fig.~\ref{fig_sim1}) it is possible to notice differences between the result obtained by the classical Jain's index and our new players fairness (PF) index. It is worth recalling that the Jain's index has value 1 when the allocation is proportional while the PF index is 1 when the allocation is the mood value.  
We can notice that the PF index considers the MMF allocation as a fair one when the available resource is small (high congestion), i.e., when there are many greedy users. In fact the MMF allocation and the mood value, are close: in such cases, both have the property of treating equally the greedy claimant, giving them the same portion of resource, independently of their demands. Instead, when $E$ increases, the MMF one is not fair any longer because it satisfies more the two users with less claim while it gives the minimal right to the one with bigger claim; in fact, in such cases the mood value becomes closer to the Proportional allocation, to the Shapley value and to the Nucleolus. The similarity between the Proportional allocation and the mood value is due to the fact that the correct way to measure the satisfaction of moderate players is through the DFS rate and increasing $ E  $ the number of moderate players increases. It follows that the allocation equalizing the DFS rates, i.e., the Proportional one, is close to the one equalizing the PS rates of user, i.e., the Mood Value. 

With 5 users (Fig.~\ref{fig_simrand} and Fig.~\ref{fig_sim2}), we can notice that, due to the fact that the group \textsc{Mg} is little, when $ E $ reach the $ 40\%  $ of the global demand the Proportional is equal to the Mood value and it shows a PF index equal to 1. Furthermore the MMF allocation show again a high (PF) fairness when the resource is little ($ 5\% $), i.e., when there are many greedy players.

\begin{figure*}[t]
	\centering
	\vspace{-0.6cm}
	\subfloat[4 users,  $E$/demand = 5\%]{\includegraphics[height=3.25cm]{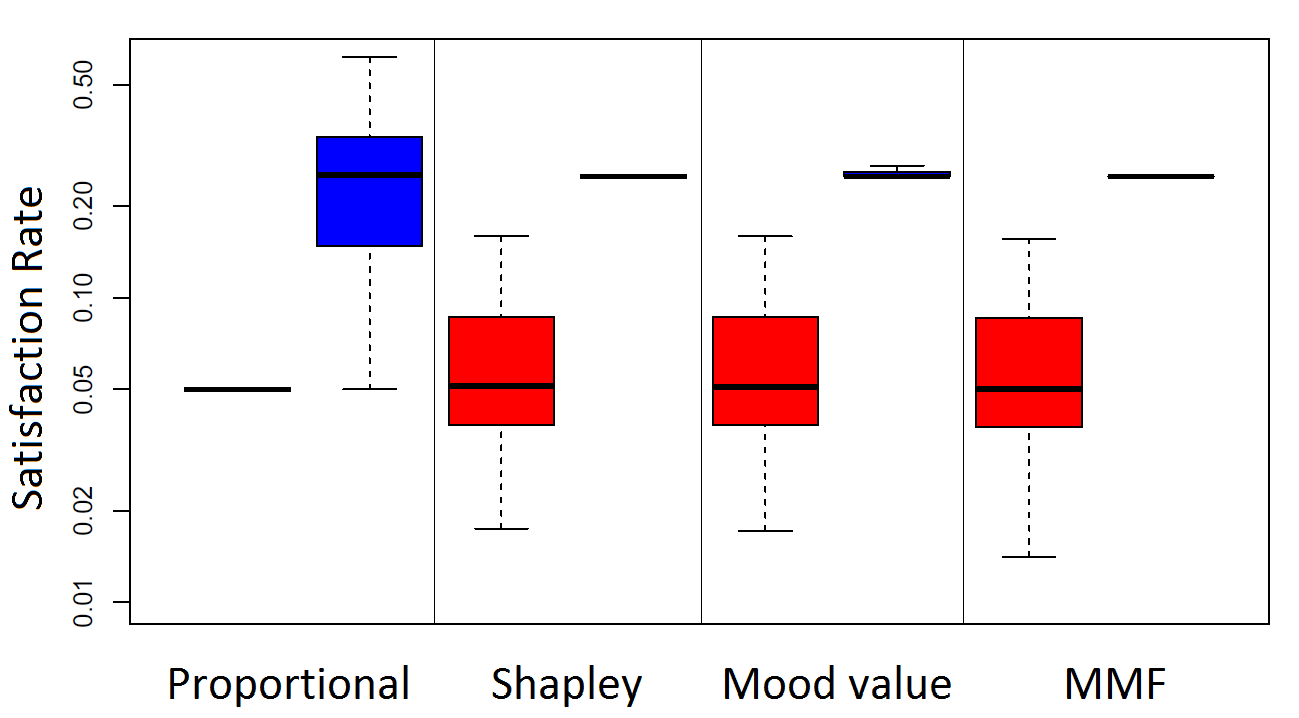}%
		\label{sodd_4}}
	\hspace{-0.1cm}
	\subfloat[8 users, $E$/demand = 5\%]{\includegraphics[height=3.25cm]{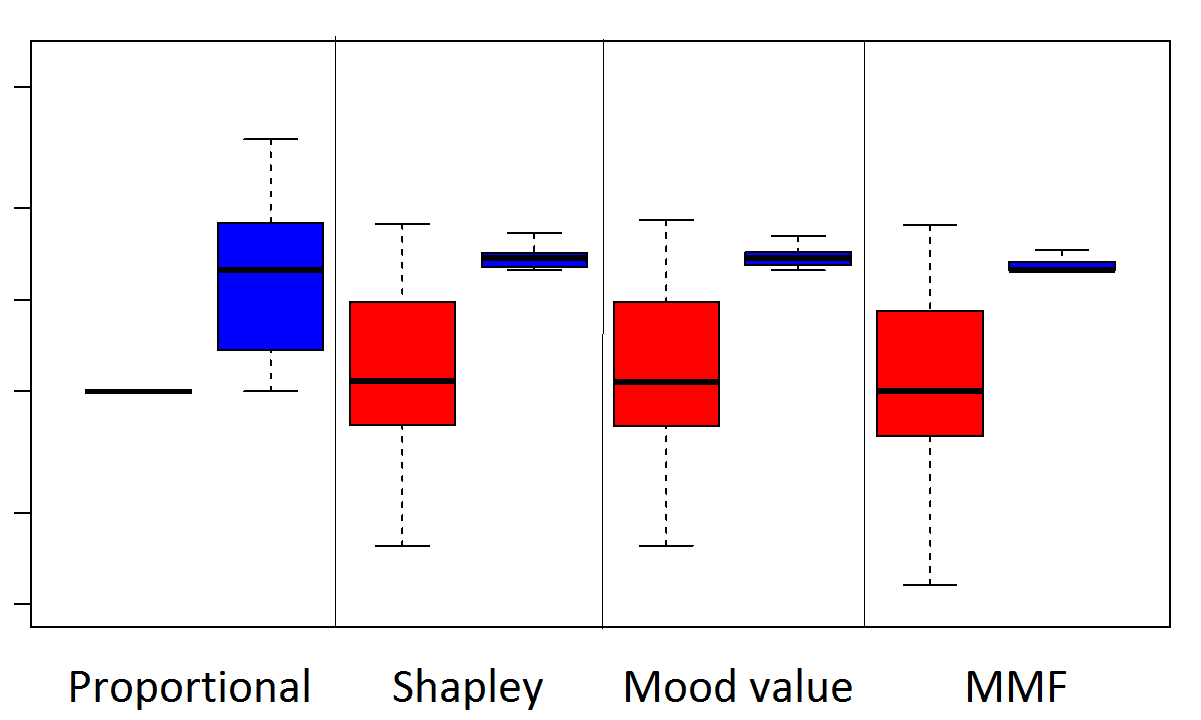}%
		\label{sodd_8}}
	\hspace{-0.1cm}
	\subfloat[16 users, $E$/demand = 5\%]{\includegraphics[height=3.25cm]{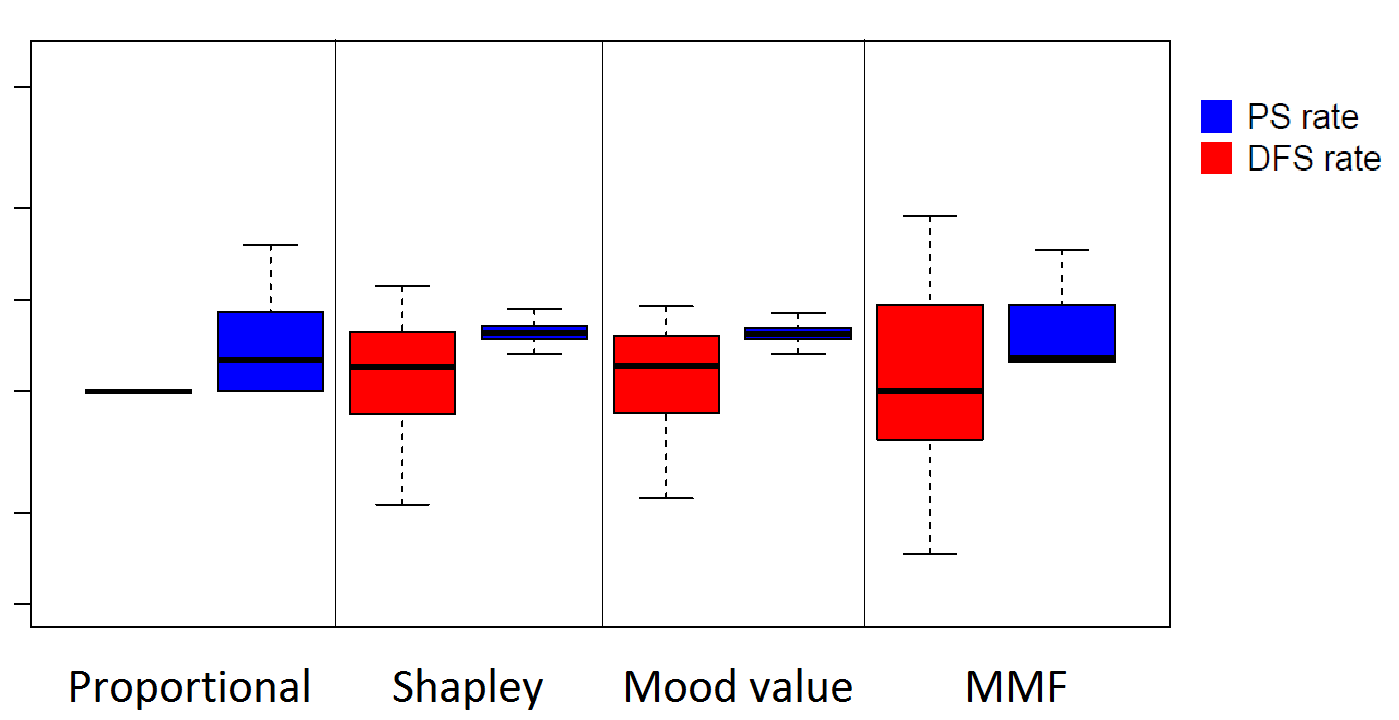}%
		\label{sodd_16}}
      \vspace{-0.05cm}  	
\subfloat[4 users, $E$/demand = 95\%]{\includegraphics[height=3.25cm]{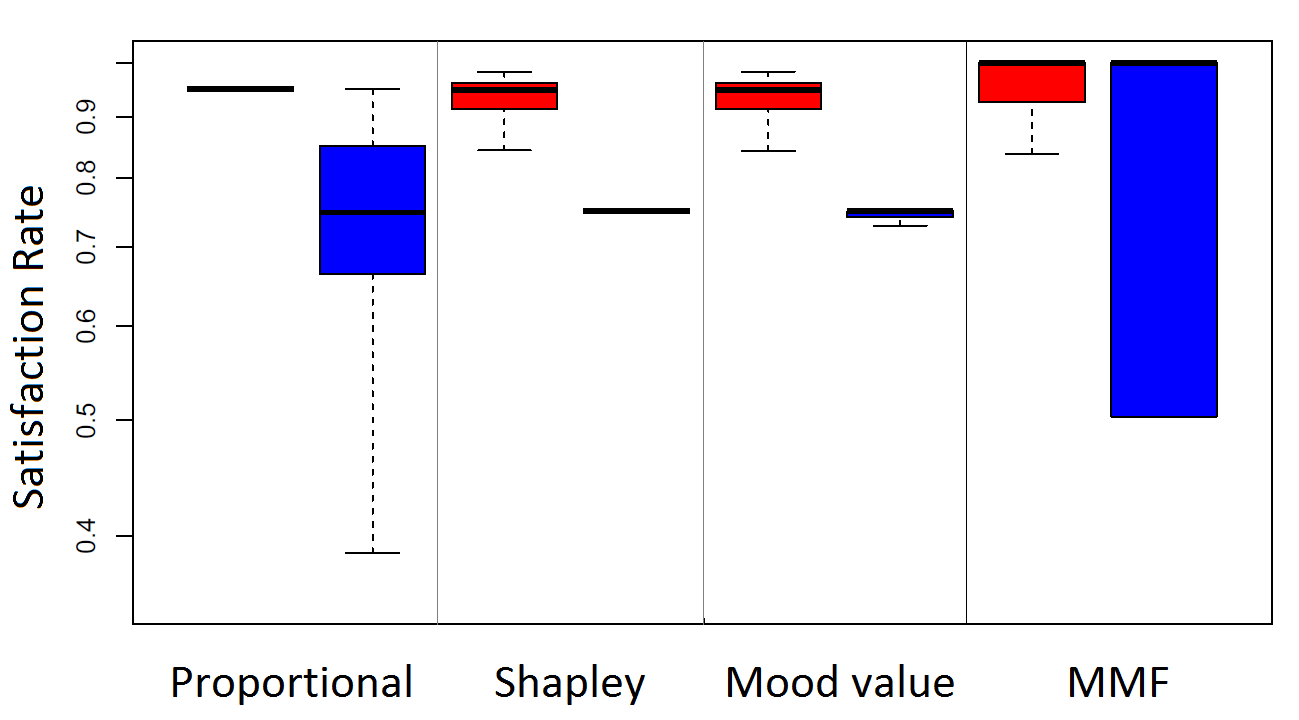}%
		\label{sodd_4_h}}
	\hspace{-0.1cm}
	\subfloat[8 users, $E$/demand = 95\%]{\includegraphics[height=3.25cm]{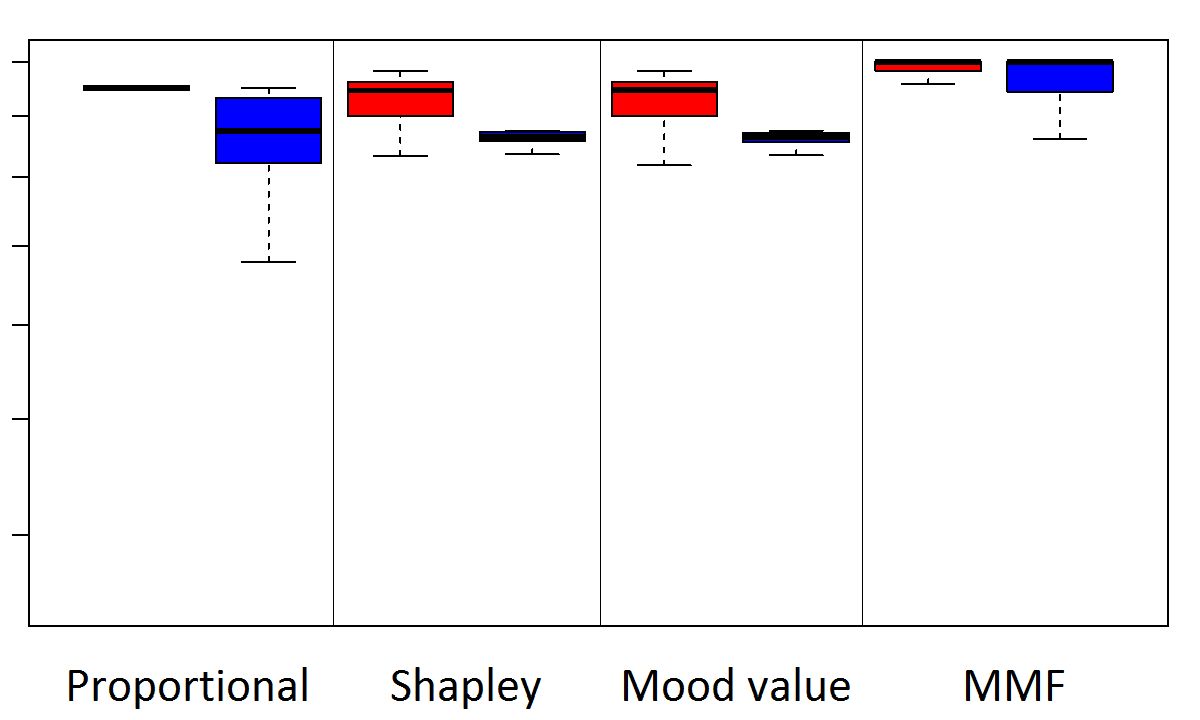}%
		\label{sodd_8_h}}
	\hspace{-0.1cm}
	\subfloat[16 users, $E$/demand = 95\%]{\includegraphics[height=3.25cm]{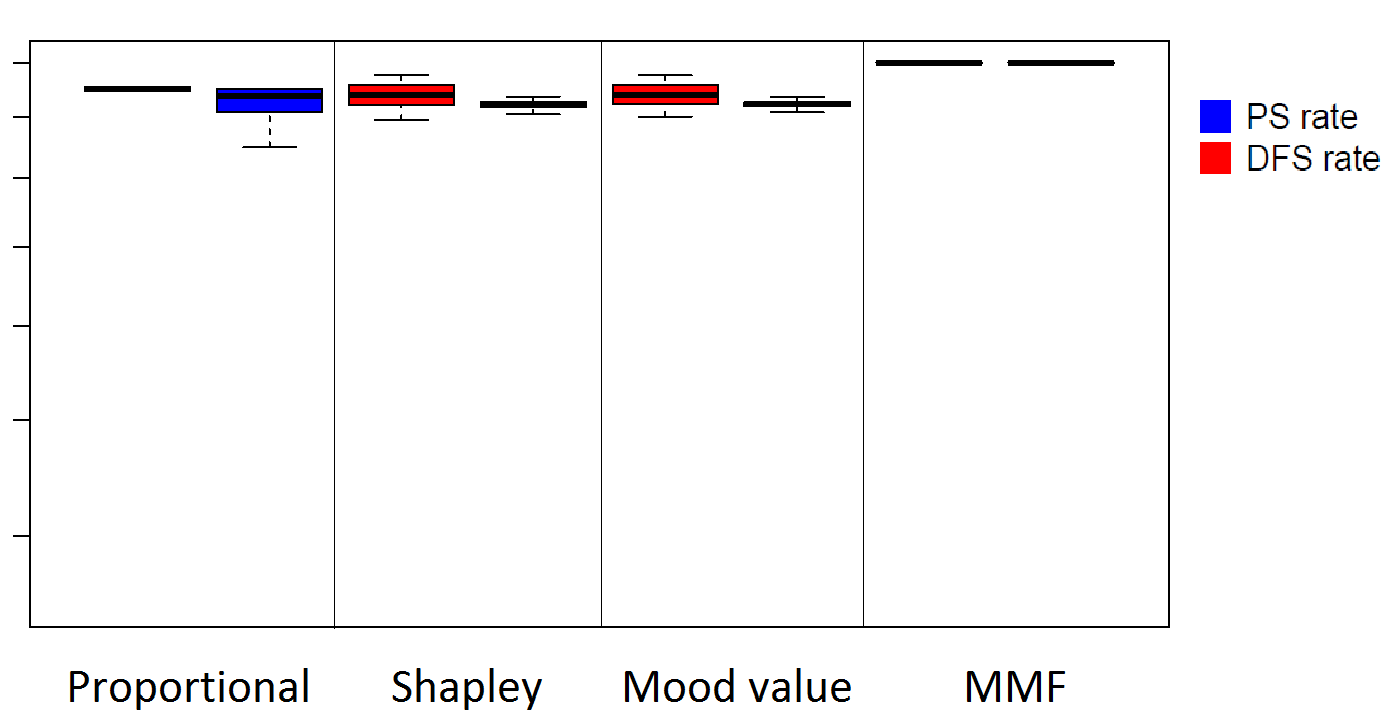}%
		\label{sodd_16_h}}
		\caption{Users satisfaction rates - $E$ over global demand = 5\% and 95\% - uniform demand distribution}
        \vspace{-0.5cm}
		\label{rateslow}
	\end{figure*}
    
With a second round of simulations we want to see how the Demand Fraction (DFS) and Player (PS) Satisfaction rates are distributed with an increasing number of users. Results are shown in Figure~\ref{rateslow} as box-plots (minimum, quantiles and maximum, without outliers), for two regimes (high congestion of a 5\% $E$/demand ratio, and low congestion of a 95\% ratio). We can notice that the value of the satisfaction is low when the resource is small (Fig..~\ref{rateslow}a,b,c) while it is higher when the congestion is low (Fig.~\ref{rateslow}d,e,f).
With 4 users, in terms of user satisfaction rate, the mood value is close to the MMF allocation when the congestion is high, and to the proportional allocation when the congestion is low. As the number of players grows, the absolute difference between allocation in terms of distribution of the satisfaction decrease. In both the congestion situations ($E$ equal to 5\% and 95\% of the demand) the Shapley value is the closest allocation to the mood value in terms of PS rate.

%
%

Summarizing, the simulations show that the Mood Value is able to nicely weight the nature (greedy or moderate) of users and of user groups. In particular it is close to the MMF allocation when the resource is scarce and to the proportional allocation when the resource is close to the global demand. Furthermore, it is worth noticing that with respect to classical game-theoretical allocation rules (Shapley Value, Nucleolus), the results show that the Mood Value shows a similar good behavior in terms of fairness, with the key advantage of having a much lower computation time complexity.  


\section{Conclusion}
\label{concl}

We proposed a game-theoretical approach to analyze and solve resource allocation problems, going beyond classical approaches that do not explore the setting where users can be aware of other users' demand and the available resource.

In particular, we proposed a new way of quantifying the user satisfaction and a new fairness index as enhancement of the Jain's index, describing and comparing their mathematical properties in detail. Accordingly to these new concepts, we propose a new resource allocation rule that meets the goal of providing the fairest resource allocation, we called the \lq Mood Value\rq, which we position with respect to game theory metrics as well the common theory of fair allocation in networks.
Finally, we test our ideas via numerical simulations of representative demand distributions, showing at which extent the mood value can approach and differ from max-min-fairness, weighted proportional, constrained-equal loss, Shapley Value and Nucleolus allocations.

\section*{Acknowledgment}
This work was partially funded by the FED4PMR \lq investissement d'avenir\rq \ project. The authors would like to thank Sahar Hoteit for her suggestions, and Deep Medhi and Catherine Rosenberg for their useful feedback.
\label{concl}


\end{document}